\documentclass[letterpaper,twocolumn,10pt]{article}
\usepackage{usenix-2020-09}
\usepackage{
			authblk, %
			url,
            xspace,
            listings,
            tikz,
            amsfonts,
            amsmath,
            amsthm,
            paralist,
            subcaption,
            courier,
            hyperref,
            cleveref,
            hhline,
            amsthm,
            textcomp,
            color,
            colortbl,
            multirow}
\usepackage{bigints}

\usepackage{algorithm}
\usepackage{algorithmic}
\usepackage{hyperref}
\usepackage[skins,breakable]{tcolorbox}

\newtcolorbox{tbox}[1][]{enhanced jigsaw,breakable,pad at break=1mm, oversize,interior hidden,colframe=black,boxrule=0.5pt,nobeforeafter=,#1}

\def\A{\ensuremath{\mathcal{A}}}
\def\Sim{\ensuremath{\mathcal{S}}}
\def\D{\ensuremath{\mathcal{D}}}
\def\H{\ensuremath{\mathcal{H}}}
\def\F{\ensuremath{\mathcal{F}}}

\def\X{\ensuremath{\mathcal{X}}}
\def\Y{\ensuremath{\mathcal{Y}}}

\def\ideal{\textsc{ideal}}
\def\ppt{\textsc{ppt}}

\newcommand{\bin}{\ensuremath{\{0,1\}}}
\newcommand{\sample}{\hskip2.3pt{\gets\!\!\mbox{\tiny${\$}$\normalsize}}\,}
\newcommand{\secparam}{\ensuremath{1^\kappa}}

\newcommand{\share}{\ensuremath{\mathsf{Share}}}
\newcommand{\recon}{\ensuremath{\mathsf{Rec}}}

\newcommand{\sh}{\ensuremath{\mathsf{sh}}}

\crefname{section}{\S\!}{\S\S}
\crefname{appendix}{\S\!}{\S\S}
\crefname{figure}{Figure}{Figures}
\crefname{table}{Table}{Tables}
\newtheorem{definition}{Definition}

\newcommand{\sys}{Co\-Vault}

\newcommand{\texp}{\mbox{$t_e$}} %

\renewcommand\paragraph[1]{\smallskip\noindent\textbf{#1}}

\usepackage{etoolbox}
\appto\UrlBreaks{\do\-}

\newif\iftr
\trtrue

\newcommand{\tr}[1]{{\color{black} {#1}}}

\newtheorem{thm}[]{Theorem}
\newtheorem{lem}[]{Lemma}

\usepackage[available, functional, reproduced]{usenixbadges}

\begin{document}

\date{}

\title{\sys: Secure, Scalable Analytics of Personal Data}

\newcommand{\sharedfootnote}{\footnotemark[1]}

\author[1]{Roberta De Viti}
\author[1\thanks{Affiliation at time work was done.}\setcounter{footnote}{1}]{Isaac Sheff}
\author[2,4\textsuperscript{\textcolor{green!80!black}{*}}]{Noemi Glaeser}
\author[3\textsuperscript{\textcolor{green!80!black}{*}}]{Baltasar Dinis}
\author[3]{Rodrigo Rodrigues}
\author[4]{\\Bobby Bhattacharjee}
\author[5\textsuperscript{\textcolor{green!80!black}{*}}]{Anwar Hithnawi}
\author[1]{Deepak Garg}
\author[1]{Peter Druschel}
\affil[1]{Max Planck Institute for Software Systems (MPI-SWS), Saarland Informatics
	Campus}
\affil[2]{Max Planck Institute for Security and Privacy (MPI-SP)}
\affil[3]{Instituto Superior T\'{e}cnico (ULisboa), INESC-ID}
\affil[4]{University of Maryland}
\affil[5]{ETH Zürich}

\maketitle
\pagestyle{plain}

\begin{abstract}
	There is growing awareness that the analysis of personal data, such as
individuals' mobility, financial, and health data, can provide
significant benefits to society. However, liberal societies have so far
refrained from such analytics, arguably due to the lack of
secure analytics platforms that scale to billions of records
  while operating in a very strong threat model. We contend that one
fundamental gap here is the lack of an architecture that can scale
(actively-)secure multi-party computation (MPC) horizontally without
weakening security. To bridge this gap, we present \sys, an analytics
platform that leverages server-aided MPC and trusted execution
environments (TEEs) to colocate the MPC parties in a single datacenter
without reducing security, and scales MPC horizontally to the
datacenter's available resources. \sys\ scales well empirically. For
example, \sys\ can scale the DualEx 2PC protocol~\cite{dualex} to
perform epidemic analytics for a country of 80M people (about 11.85B
data records/day) on a continuous basis using one core pair for every
30,000 people.

\end{abstract}

\section{Introduction}
\label{sec:intro}

Data about individuals' health, diet, activity, mobility, social
contacts, and finances is being captured at high resolution by smart
devices and apps, by online services, and by offline services such as
hospitals and banks that store customer data. Large-scale {\em
  analysis} of this data could benefit (i) public health, e.g., by
studying the spread of epidemics with high spatio-temporal resolution,
or new and rare diseases; (ii) sustainability, e.g., by informing
transportation, energy, and urban planning; (iii) social welfare,
e.g., by uncovering disparities in income and access to education or
services; and (iv) economic stability, e.g., by providing insight into
markets, among many others.

Yet, liberal societies have mostly refrained from making personal data
available for large-scale analysis, even if doing so would clearly be
in the public interest (e.g., for scientific research). Arguably, this
is largely out of concern that data leaks and misuse would harm
individuals and businesses, erode the public's trust and deter
voluntary data contributions. There is a growing realization, however,
that this conservative approach conflicts with urgent societal
challenges like public health and sustainability, where savings on the
order of billions of Euros are thought possible using data-driven
innovation in the EU alone~\cite{european-data-strategy,
  digital-strategy}. To this end, the EU Parliament approved the Data
Governance Act (DGA)~\cite{dga}, which seeks to facilitate voluntary
contributions of high-resolution data from diverse \emph{data
  sources}---individuals, private companies and public bodies---for
analysis by authorized parties (e.g., scientists, government agencies)
in the EU. The DGA also defines {\em data intermediaries} who act as
data clearinghouses and facilitate analytics by authorized parties
securely.

Given the sensitivity and volume of the data in question, we believe
that the following properties are fundamental to any data intermediary
or similar secure analytics platform: (1) \emph{Data confidentiality}
at all times -- while data is \emph{in use}, at rest and in
transit. Confidentiality should \emph{not rely on a single root of
  trust for any component (even hardware)} because strong adversaries
who can compromise or backdoor software and hardware are fathomable
in this context, (2)~\emph{Scalability}: The platform should scale to
datasets spanning entire nations or continents (hundreds of billions
of records) and millions of data sources, (3) \emph{Integrity} of the
analytics results, which is important when analytics inform
political, economic, or healthcare decisions, and (4) \emph{Selective
  forward consent (SFC)}: Each data source should be able to decide
which analysis queries and which analysts can use its data, and set an
expiration time on the data.

To the best of our knowledge, the design of a platform that satisfies
just the first three properties together remains an unsolved
problem. By definition, data confidentiality without a single root of
trust requires secure multi-party computation
(MPC)~\cite{Yao1982,DBLP:conf/stoc/GoldreichMW87,cramer2015secure,katz2020modern,DBLP:conf/stoc/Ben-OrGW88}. However,
owing to their onerous computation and bandwidth overheads, MPC-based
systems have never been shown to scale to hundreds of billions of
records. Even the best contemporary MPC systems scale only to millions
of
records~\cite{Nayak2015,Chowdhury20,190910,stealthdb,DBLP:conf/ccs/OhrimenkoCFGKS15,papadimitriou2016,Bater2017,Volgushev2019,secureml},
and that too only with passive security, where even compromised
parties are assumed to follow the protocol -- an unrealistic assumption
in our setting. Instead, we require MPC with active security,
which ensures confidentiality and integrity even when compromised
parties behave arbitrarily. In the active setting, we are not aware of
any MPC system that scales even to millions of records.

Accordingly, this paper presents \sys, a new system that scales any
underlying MPC protocol horizontally to the resources of a
datacenter. We do not invent sophisticated new protocols. Rather, we
observe that a careful combination of \emph{existing} distributed
scaling techniques and hardware features suffices to scale MPC
horizontally. Furthermore, we observe that these same techniques 
also enforce integrity of results and SFC, yielding a system that
satisfies all four properties listed above.

We validate {\sys}'s design by empirically demonstrating horizontal
scaling in the 2-party setting for two actively-secure, garbled
circuit-based MPC protocols in slightly different leakage models:
malicious~\cite{agmpc,DBLP:journals/iacr/KatzRRW18}
and malicious with a 1-bit leak (which leaks 1 bit of information for
significantly better performance)~\cite{dualex}. For comparison, we
also discuss how passively-secure MPC would scale horizontally with
{\sys}, even though this is not our intended use case. Our
experiments focus specifically on MPC using garbled circuits because
our target application, query-based analytics, is dominated by
combinatorial operations such as integer comparison, for which garbled
circuits are more efficient than other MPC
protocols~\cite{Senate21}. However, {\sys}'s design does not
depend on garbled circuits and can be used to scale any class
of MPC protocols.

Next, we briefly outline our design and how it provides all four
properties, starting with horizontal scaling.

\paragraph{Scalability: Bandwidth problem.}
The main bottleneck in scaling MPC protocols horizontally is network
bandwidth. A single pair of cores in the 2-party setting can saturate
several Gbps of network bandwidth when evaluating garbled
circuits. (Lower bandwidth increases query latency proportionally and
wastes CPU resources.)  The required bandwidth increases proportionally
with the number of cores. For example, if each party has
1,000 cores, a cross-party bandwidth of several Tbps is
needed to fully utilize the cores. Such high
bandwidth is available only inside a datacenter.

Thus, in order to scale MPC horizontally, it is necessary to colocate
the MPC parties in the same datacenter. However, colocation
increases the parties' susceptibility to compromise by the same technical, political, and
physical means, thereby negating the very purpose of using MPC.

Our first insight is that we can colocate parties in the same
datacenter, while preserving their independence to compromise, by
separating the parties' software stacks and isolating the
  parties in \emph{trusted execution environments (TEEs) of independent
  hardware vendors}, such as Intel's TDX, AMD's SEV-SNP and ARM's
CCA~\cite{intel-tdx,amdsev-github,arm-cca}.
This design places the root of trust in the infallibility of
  $(n-t)$-of-$n$ heterogeneous TEEs, i.e., it tolerates the compromise
  of up to $t$-of-$n$ TEEs (1-of-2 in our prototype). Using this idea, we can
colocate all MPC parties to obtain very high cross-party bandwidth
without weakening MPC security:
To compromise confidentiality or integrity,
an attacker would have to compromise more
vendor-independent TEEs than the underlying MPC protocol's compromise
threshold, which is an extremely high bar for any attacker.

Although the isolation of MPC parties in diverse TEEs has been
considered in concurrent work for other
purposes~\cite{dauterman2022reflections,DBLP:conf/osdi/ConnellFSDP24},
we believe that {\sys} is the first system that leverages diverse TEEs and
party colocation to harness the large bisection bandwidth and core
count of a datacenter to scale MPC analytics without weakening
the MPC threat model.

\paragraph{Scalability: Core parallelism.}
Even with the bandwidth issue resolved, we need a way to exploit
multi-core parallelism. We observe that standard scaling
techniques from the distributed systems literature can be applied even
when the underlying compute units run MPC. Specifically, \emph{we
  implement MapReduce~\cite{dean2004mapreduce} over MPC}: Given an
analytics query on a very large private dataset, we divide the query's
MPC circuit into small map and reduce circuits, which are scheduled across
the available cores at each party (as in standard MapReduce). 

We implement mappers and reducers using standard MPC circuits for
the following basic operations on lists of records: linear scans, sorting, compacting (moving
marked entries to the tail) and merging sorted lists. As we explain
later, this set of primitives suffices for all filter-groupby-aggregate queries. We carefully avoid leaks through network
side channels by padding all communication between mappers and
reducers to fixed sizes. Additionally, we support data-dependent queries through
a new and simple oblivious data retrieval (ODR) scheme that enables
lookups through indexes on private fields without revealing the
data access pattern.

We note that MapReduce can be substituted with any other distributed
scaling paradigm, as long as the paradigm's basic primitives can be
efficiently implemented in MPC circuits.

\paragraph{Summary of {\sys}'s basic design.}
At a high-level, {\sys} combines \emph{horizontal scaling}
with \emph{server-aided
  MPC}~\cite{DBLP:journals/corr/Corrigan-GibbsB17,DBLP:journals/iacr/RatheeSWP22,DBLP:conf/sicherheit/DoernerMM22,DBLP:conf/cms/BugielNSS11,DBLP:conf/sac/Kerschbaum15,JakobsenNO14}.
{\sys}'s analytics backbone consists of a fixed set of $n$
administratively independent MPC parties, each with a large number of
dedicated servers, all colocated in the same datacenter with 
high-bandwidth interconnects, and running in TEEs of $n$ independent
vendors.

Data sources -- companies, individual devices and public bodies --
contribute their data to {\sys} by \emph{secret-sharing} each datum
among the $n$ parties. The $n$ parties store these shares in local
databases, possibly after pre-processing in MPC (e.g.,
range-checks). Analytics queries are executed by the $n$ parties on
the stored shares using MPC. The MPC protocol must tolerate the
malicious compromise of a subset of the parties (up to a threshold denoted $t$),
and all MPC computations are scaled horizontally by distributing them across 
all available cores, as explained earlier.

\sys~is primarily designed for data donated by individuals, a
  setting that necessitates SFC and support for many data sources, unlike
  prior work like Senate~\cite{Senate21}, where a small number of sources
  contribute larger amounts of locally aggregated data. While
  \sys~supports this latter configuration as well, it is not its
  focus.
Next, we describe how {\sys} supports integrity of outputs and SFC.

\paragraph{Integrity of outputs.}
In our server-aided MPC setting, ensuring the integrity of outputs
simplifies to two sub-problems: (i) Input integrity: Verifying that all
MPC parties compute on \emph{unmodified} data shares
received from the data sources, and (ii) Computation integrity: Ensuring
that the MPC computation executes correctly. Both sub-problems have
standard solutions, which we leverage.
For input integrity, we use \emph{authenticated secret
  sharing}~\cite{RabinBenOr1989,damgaard2012multiparty,eskandarian2021clarion}. Each
datum is secret-shared with a MAC that is created by the datum's
source. A datum's MAC is checked inside MPC when the datum is used for
analytics, detecting any manipulation of data shares by rogue parties.
For computation integrity, we rely on actively-secure MPC
protocols~\cite{DBLP:conf/stoc/Ben-OrGW88,agmpc,DBLP:journals/iacr/KatzRRW18,dualex},
which detect any deviation from the intended protocol.

\paragraph{Selective forward consent.}
We further observe that TEEs, which are already a part of our design,
can be harnessed to enforce SFC as follows. The code of each TEE
instance (running on behalf of a MPC party) is only capable of
participating in a specific set of analytics queries, authorized by a
specific set of queriers up to a specific expiration time. The
queries, authorized queriers, and expiration time of each
instance are publicly advertised. A data source uploads data shares to
a TEE instance only if it consents to the instance's queries,
authorized queriers, and expiration time, and it can attest --
\emph{using TEE attestation} -- that the instance's code implements
these parameters. This enforces SFC.

\paragraph{Contributions.}
Our contributions include:
(a) \sys, a system that scales MPC horizontally to datacenters, fully
exploiting core parallelism. {\sys} leverages diverse TEEs to
facilitate colocation of MPC parties with very high cross-party
bandwidth, without weakening independence of the parties;
(b) Three technical components of possibly independent interest: A
provably secure authenticated secret-sharing scheme
(\cref{subsec:macs}), an actively secure 2PC protocol with a 1-bit
leak that reveals information to a third entity, not the computing
parties (\cref{subsec:dualex}), and a custom oblivious data retrieval
(ODR) scheme (\cref{subsec:database}) that supports database lookups
on private indexes in constant time;
(c) An experimental evaluation of {\sys} with epidemic analytics as an
example scenario, which shows that {\sys} can handle data for a
country of 80M people using a single core pair per 30,000 individuals.
We are not aware of any other MPC-based analytics platform that can
securely perform similarly complex queries in a comparably strong
threat model at this scale.

\section{\sys\ Overview}
\label{sec:background}

We survey relevant building blocks, sketch {\sys}'s architecture, and
provide a roadmap for the rest of the paper.

\begin{figure}[t]
\centering
\includegraphics[scale=0.5]{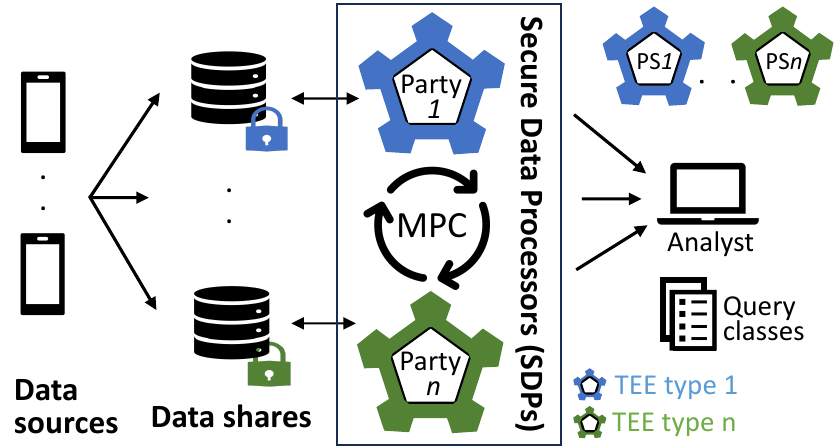}
\caption{{\sys}'s secure data processing. \emph{Conceptually}, data
  sources contribute their data in such a way that only queries within
  a given query class can be executed on their data, by analysts
  authorized by the query class, and for a period defined by the query
  class. \emph{Concretely}, (a) Each data source secret-shares its
  data into $n$ authenticated shares, encrypts each share for a
  different MPC party, checks that the parties have been verified, and
  uploads the shares to an untrusted database in a datacenter.  (b)
  Queries are processed by a set of SDPs hosted in the same
  datacenter. Each of the $n$ SDPs, one for each party, is
  encapsulated by a different TEE type and accesses its corresponding
  data share; the SDPs jointly perform MPC, and deliver their
  respective share of the query result to an authorized analyst. (c)
  Provisioning servers (PSs), one per party and implemented in the
  party's TEE type, attest and provision all SDPs of a party with
  their keys.  }
\label{fig:arch}
\end{figure}

\subsection{Building blocks}
\label{subsec:tech}

\paragraph{Secret sharing} is a method for sharing a secret value
among multiple parties, such that the value can only be reconstructed if a
sufficient fraction of shares (possibly all) are combined. {\em
  Authenticated} secret sharing adds auxiliary
information that allows the parties to verify the authenticity of a value that they have
reconstructed from their shares.

\paragraph{Secure multi-party computation (MPC)} enables multiple
parties to jointly compute a function without revealing their
respective inputs.  MPC protocols are either passively secure (corrupt
parties are semi-honest~\cite{Yao1982,huang2011faster}) or actively
secure (corrupt parties may be malicious~\cite{malicious2pc}).
In $t$-of-$n$ active security, up to $t$ of the $n$ parties may be
compromised. Depending on the protocol, $t$ varies from $n/3$ to
$n-1$.

Garbled circuits (GCs) are a specific class of secure 2-party
  computation protocols (2PC) ($n=2$,
$t=1$)~\cite{2pc_is_practical}. GCs exist with both passive and active
security. A specific GC protocol, DualEx~\cite{dualex}, achieves
  active security apart from leaking 1 bit of the input, in exchange
  for achieving performance close to that of passive
  security. Techniques exist both to restrict the predicate the
  adversary can leak and the probability that a leakage
  occurs~\cite{10.1007/978-3-662-46494-6_11}. All GCs are data
oblivious because circuits have no control flow.

\paragraph{Trusted execution environments (TEEs)} are supported by
several recent general-purpose CPUs, e.g., Intel SGX, Intel TDX,
AMD SEV and ARM CCA~\cite{sgx_explained,amdsev-snp,intel-tdx,arm-cca}.
TEEs provide confidentiality and integrity of data and computation
under a threat model that tolerates compromised operating systems,
hypervisors, and even some physical
attacks~\cite{tz_sok,sgx_explained}. Newer TEEs -- SEV-SNP, TDX, and
CCA -- encapsulate an entire VM, providing easy software migration.
An \emph{attestation protocol} verifies that a VM executes in an
authentic TEE of a given type and that its initial memory state (code
and data) matches an expected secure hash value (measurement). Upon
successful attestation, the VM is provided with cryptographic material
needed to authenticate itself to remote parties and to access data
sealed for it.  The security of a TEE depends on the integrity of its
vendor's certificate chain, proprietary hardware, and firmware.

\subsection{\sys\ roadmap and threat model}
\label{subsec:threat-model}

We introduce {\sys}'s secure data processing component (SDP) (see
\cref{fig:arch}) in \cref{sec:fe-construction}.  In
\cref{sec:covault_design}, we show how to scale \sys's query
processing horizontally with MapReduce~\cite{dean2004mapreduce}, where
each mapper and reducer is an instance of the SDP. In
~\cref{sec:full_construction_design}, we present our prototype
SDP implementation, which relies on 2PC with garbled circuits.
In \cref{sec:evaluation}, we evaluate \sys\ in the context of epidemic
analytics as an example application scenario, and compare its
performance to other systems with a similar threat model.
\iftr
The appendix provides additional details on \sys's hardware and software
attestation process (\cref{sec:community-approval}), oblivious
map-reduce technique (\cref{app:details-map-reduce}), data ingress
processing (\cref{sec:ingress}), and the epidemic analytics scenario
(\cref{sec:epidemics}), along with supplementary evaluation material.
Furthermore, the appendix includes our algorithm for estimating
end-to-end query latency (\cref{sec:mr-extrapolation}) and security 
proofs for the primitives used in \sys\ (\cref{app:proofs}).
\else
A technical report (TR) \cite{covault-tr} provides formal security proofs
for the primitives used in {\sys}, and other details.
\fi

\paragraph{Threat model.}
\sys\ relies on an underlying MPC protocol with $t$-of-$n$ security in
a deployment-chosen threat model (our intended use-case is protocols
with active security or active security with fixed bit-leakage in
exchange for significantly better performance). {\sys} inherits the
adversary model from the underlying MPC protocol used by the SDPs.
{\sys} places the code of the $n$ parties inside the TEE technologies
of $n$ different vendors. Consequently, we assume that $n-t$ of these
$n$ TEE technologies/vendors are uncompromised; the remaining $t$ TEE
technologies/vendors may be compromised (statically) using any means
allowed under the MPC protocol's adversa\-rial model.

\sys\ assumes that the software executed by the parties inside TEEs is
correct and the TEEs are properly attested. One approach is to
delegate these tasks---which have to be performed by any secure system
that relies on TEEs---to a community of experts as described in \iftr
\cref{sec:community-approval}.  \else the TR~\cite{covault-tr}.  \fi

We assume that data sources follow the data contribution protocol and
do not deliberately contribute false or biased data (data poisoning
attacks). A malicious source can compromise the confidentiality and
SFC of its own data (but not of others' data), and bias query results
relying on its data to the extent that the query is sensitive to that
source's data. Defenses against data poisoning like integrity,
consistency, and plausibility checks are orthogonal to this work and can
be integrated into \sys's data ingress processing.

We make standard assumptions about the security of cryptographic
primitives used in \sys's design, and assume a PTIME
  adversary. We assume that data is stored in \sys\ for periods short
enough that the keys remain secure despite advances in cryptanalysis
and compute power.  Denial-of-service attacks are out of scope; they
can be addressed with orthogonal techniques.

TEEs have known vulnerabilities including root-key leaks, attacks on remote
attestation, vendor compromise, conventional side-channel attacks
(memory-access patterns, CPU cache timing), microarchitectural
side-channel attacks, and integrity attacks (power-based, Rowhammer)~\cite{li2024sok, tz_sok}.
Nonetheless, our approach raises the bar
for vendor compromise and implementation-specific attacks, and
improves security over TEE+MPC setups that place the root of trust in
a single TEE vendor (e.g.,~\cite{cryptflow}).  \sys~is immune to conventional
side-channel attacks as the underlying MPC protocols are data
oblivious. Software mitigations of microarchitectural attacks—to the
extent that they exist—can be applied inside TEEs. Consequently, a
successful attack on \sys~requires the compromise of TEEs from $t+1$
distinct vendors running hardened code without conventional side
channels. We believe that success with such attacks is unlikely, but
they are possible in principle, and mitigating them is outside the
scope of this paper.

\paragraph{Deployment considerations with TEEs.}
TEE adoption for commercial use is growing rapidly. Two leading Cloud
providers—Google Cloud and Azure—offer both AMD SEV-SNP and Intel TDX
TEEs with high-bandwidth interconnects~\cite{azureopt, gceopt}.
These TEEs encapsulate entire VMs and the Cloud providers support
transparent VM deployment. Hence, the additional deployment overhead is
limited to code attestation (for each TEE vendor). To mitigate this burden,
Cloud providers offer client libraries, code wrappers and detailed tutorials~\cite{gceatt, azureatt}.

\section{Secure data processing}
\label{sec:fe-construction}

In this section, we present the design of {\sys}'s secure data
  processing, beginning with the API and then incrementally refining
a strawman design.

\label{subsec:api}

\begin{figure*}
  \small
  \center
  \begin{tabular}{|p{\textwidth}|}
    \hline
        {$MPK[Q, D, \texp] \leftarrow {\bf Setup}(1^\kappa, Q, D, \texp)$} \\
        \begin{tabular}{@{}p{\textwidth}@{}}
          MPC to initialize a new query class and produce a
          public-private key pair $(MPK[Q, D, \texp], MSK[Q, D,
            \texp])$ bound to the class $[Q, D, \texp]$. $MPK[Q, D,
            \texp]$ is published for data sources, $MSK[Q, D, \texp]$
          is retained by the $n$ parties.  \textbf{Inputs:} $\kappa$:
          security parameter, $Q$: set of allowed queries, $D$: set of authorized analysts,
          $\texp$: expiration time.
            \textbf{Outputs:} $n$-element public key $(MPK[Q, D, \texp])$. 
          \end{tabular} \\
    \hline\hline
        {$C \leftarrow {\bf Contribute}(MPK[Q, D, \texp], m)$} \\
        \begin{tabular}{@{}p{0.95\textwidth}@{}}
        Executed locally at a data source to contribute data $m$ to the query class $[Q, D, t_e]$.
        \textbf{Inputs:} $MPK[Q,D, \texp]$: public key, $m$:
        data to encrypt. \textbf{Outputs:} $C$: $n$ encrypted shares of $m$.
        \end{tabular} \\
        \hline\hline
    $q(m_1,\ldots,m_k) \leftarrow {\bf Query}(q, d, C_1,\ldots,C_k)$ \\
          \begin{tabular}{@{}p{0.95\textwidth}@{}}
            MPC to compute $q(m_1,\ldots,m_k)$. If the caller is $d$,  $C_i$s are encrypted shares of $m_i$s in the same query class $[Q,D,\texp]$,  $d \in D$, $q \in Q$, and current
            time $<$ $\texp$, then return $q(m_1,\dots,m_k)$, else return $\bot$.
            \textbf{Inputs:}  $q$: query, $d$: analyst, $C_i$: $n$ encrypted shares of data from source $i$.
            \textbf{Outputs:} $q(m_1,\ldots,m_k)$ or $\bot$.
          \end{tabular} \\
   \hline
  \end{tabular}
  \caption{\sys's API ($MPK$, $MSK$, $C$ are $n$-element
      vectors, where $n$ is the number of parties.)}
\label{fig:fe-api}
\end{figure*}

\sys's API is shown in \cref{fig:fe-api}.  Analysts who wish to
solicit data contributions define a {\em query class} using the \textit{Setup}
operation. A query class is defined by the triple $[Q, D, \texp]$,
where $Q$ denotes a set of queries, $D$ a set of authorized analysts,
and $\texp$ a time at which all contributed data should expire and no
longer be available for analysis.  A data source contributes data to a
query class only if it is comfortable with the class's $Q$, $D$ and
$\texp$.  Data sources contribute data to a query class using the
$Contribute$ operation, and authorized analysts may execute queries
using the $Query$ operation.

We aim to realize this API while satisfying the four
properties listed in~\cref{{sec:intro}}:\\
\textbf{1) Confidentiality.}  Contributed data remains confidential
at all times, except for releases allowed by SFC (see below).\\
\textbf{2) Scalability.}
All data processing can be horizontally scaled out within a
single datacenter, without weakening our threat model.\\
\textbf{3) Integrity.}  Any modifications to the data shares or the query
result is detected.\\
 \textbf{4) SFC.} A data source contributes its data to specific
  query
  classes. That data is used in accordance with those query classes' $[Q, D, \texp]$ triples.\\
  
Next, we describe a strawman design (S1) for 
core SDP functionality, which attains confidentiality, integrity, and
  SFC in a weak threat model.  In \cref{subsec:s3}, we refine this
strawman to  our full construction (S2), which satisfies the
  desired strong threat model even though all parties execute in the
  same datacenter. In \cref{sec:covault_design}, we show how to
  exploit this colocation to achieve the remaining property,
  scalability.

\subsection{Strawman (S1): Server-aided MPC}
\label{subsec:s1}

Our first strawman, S1, implements the API (Figure~\ref{fig:fe-api})
by combining secret sharing and $n$-party MPC (but not TEEs,
  which we add in \cref{subsec:s3}). Specifically, $n$ independent
parties \emph{jointly} hold shares of the secret key, $MSK[Q, D,
  \texp]$.  S1 attains confidentiality, integrity, and SFC, but
in a weak threat model that assumes that parties have independent
  roots of trust despite being co-located in the same
  datacenter.

\paragraph{Components.}
 To contribute its data, a source uses a $n$-of-$n$ {\em
    authenticated secret-sharing scheme} (the cleartext cannot be
  recovered unless all $n$ shares are available); then, it entrusts
  each data share to a different one of $n$ parties by encrypting the
  share with the respective party's public key.  To execute a query,
  the parties decrypt their shares of data locally, and run $n$-party
  MPC on their shares to reconstruct the uploaded plaintexts
  $m_1,\ldots,m_k$ and compute the query result, which is provided to
  the analyst. The chosen MPC scheme must be actively secure with
$t$-of-$n$ static compromise: no party learns anything about
$m_1,\ldots,m_k$
as long as at least $n-t$ parties are honest.  Furthermore, the chosen
MPC scheme is data oblivious, i.e., without any control flow.  A
simple way to attain data obliviousness is to use circuit-based MPC.

\paragraph{Implementation.} Next, we sketch how S1 implements the
$n$-party API shown in Figure~\ref{fig:fe-api}, which can be used
after the $n$ parties are initialized.

\smallskip
\noindent
{${\bf Setup}$: } Each party executes this function locally, produces
a standard asymmetric public-private key pair, keeps its private key
locally alongside $Q$, $D$ and $\texp$, and publishes the public
key. MPK and MSK denote $n$-element vectors (one element for each
party/share) of the public and private keys, respectively.

\smallskip
\noindent
${\bf Contribute}$: A data source runs this function locally, secret-shares
its data $m$ into $n$ shares with an authenticated secret-sharing
scheme, and encrypts each share with the corresponding party's public
key in the vector $MPK[Q,D, \texp]$.  $C$ is the vector of the
encrypted shares.

\smallskip
\noindent
${\bf Query}$: The analyst $d$ calls this function separately on each
party, authenticates itself as $d$, and provides the desired query
$q$ along with the party's encrypted shares of the vectors
$C_1,\ldots,C_k$. If each party can authenticate $d$, finds that
  $d\in D$, $q \in Q$ and its local time is less than $\texp$, then
  the parties decrypt their respective encrypted shares locally and
  run MPC. The MPC computation first checks the authenticity of the shares
  (via the authenticated secret-sharing scheme), and then computes
  the query result (which is $\bot$ if any check fails).
  This result is encrypted for the analyst and signed by a private
  key that exists only in shared form.

\paragraph{Properties.}
S1 has the desired confidentiality, integrity, and SFC properties
  assuming that at least $n-t$ parties are honest. Briefly,
  confidentiality results from source secret-sharing plus
  MPC; integrity from MACs plus malicious-mode MPC; and SFC from TEEs
  plus attestation. S1 satisfies our threat model in theory.  In
  practice, realizing the $n-t$ honest-party assumption of MPC
requires the parties to have independent roots of trust (for
compromise), i.e, their compromisability should not be
correlated. This independence is questionable when parties are
  colocated in the same datacenter without an additional mechanism to
  separate their roots of trust. Since we want colocation for scalability, 
  we introduce such a separation mechanism on top of S1 
  in our next construction (S2).

\subsection{Full construction (S2): S1 + TEEs}
\label{subsec:s3}

Our full construction, S2, extends S1 by isolating
  the (colocated) parties from each other using diverse TEEs.
  
\paragraph{Components.}
S2 executes each of the $n$ parties inside a TEE of a different,
independent design and implementation. The code of each party is
verified to be correct and the corresponding initial measurement of
each party's TEE is attested.

\paragraph{Implementation.}
During initialization, the TEEs for each of the $n$ parties are
  started and attested.  All subsequent
  communication among and with the TEEs occurs via secure channels
  tied to the attestation. The API functions from S1 execute in
  TEEs. In the API call \textit{Contribute}, the user encrypts its
  data shares only if the TEEs have been correctly attested.

\paragraph{Properties.}
Like S1, S2 provides the desired properties of confidentiality,
  integrity, and SFC within our threat model.
However, S2 goes beyond S1, where colocation of the $n$ parties correlates
the parties'\ compromisability: in S2, the $n$ attested TEEs provide
independent roots of trust for the parties even when
the parties are colocated in the same datacenter.

\subsection{SDP construction details}

Next, we describe the SDP construction, which is based on S2,
  in more detail.

\paragraph{General setting.}
The $n$ parties execute in different, independent TEE
implementations. Each party consists of a
single provisioning service (PS) and one or more secure data
  processors (SDPs).  Each PS acts on behalf of its
party, attesting and provisioning the SDPs in its pipeline with the
key necessary to decrypt data shares.

SDPs are the components that perform MPC: each party's SDPs
\emph{collectively} form one \textit{party} in the MPC protocol.  A
corresponding set of SDPs (one from each party) handles one query
class $[Q, D, \texp]$, but the same query class may be handled by
several sets of SDPs (for scaling).

Each party's PS holds information about defined query classes.  For
each class $[Q, D, t_e]$, this information consists of the measurement
hashes of the TEEs (one per TEE type/party) that jointly implement $Q$
in MPC, the public keys of the authorized analysts in $D$, the data
expiration time $t_e$, and the public keys of the SDPs that implement
that query class.

PSs and SDPs can be safely shut down and re-started from their sealed
state~\cite{intel-sgx-sealing} without re-attestation, and the sealed
state can be replicated for persistence. We discuss database rollback
prevention in \cref{subsec:database}.

\paragraph{System initialization.}
Each party instantiates its PS; then, each PS generates its party’s
private-public key pair.

\paragraph{API call ${\bf Setup}(1^\kappa, Q, D, \texp)$}:
\label{subsec:attestation}
\label{subsec:setup}
A new query class can be created by interested analysts using the
\textit{Setup} call. Each party spawns a fresh SDP for the query class,
and configures it with the class parameters $[Q,D,\texp]$.

The SDP is remotely attested by its PS and provisioned with
cryptographic keys, including the SDP's secret key to decrypt its
data shares (in the sense of standard asymmetric
cryptography). Each PS stores the query class $[Q,D,\texp]$ and the
public keys of the $n$ SDPs, and advertises them publicly. The secret
and public keys of the SDPs are, respectively, the vectors
$MSK[Q,D,\texp]$ and $MPK[Q,D,\texp]$ of the API.

\paragraph{API call ${\bf Contribute}(MPK[Q, D, \texp], m)$}:
\label{subsec:secret-sharing}
Data sources contact the PSs to retrieve information about the available
query classes $[Q,D,\texp]$, as well as the public keys of the SDP
sets that implement them. Before contributing data to a query class $[Q, D,
  \texp]$, a data source verifies that the SDPs implementing
that class have been attested.

Data sources share data using an \emph{authenticated secret-sharing
scheme}, which allows SDPs to verify that the shares have not been
modified before they use the sharing (input integrity).  The data
source uses the authenticated secret-sharing scheme to create $n$
shares of its data, and encrypts each share with the public key
of a different member of the set of SDPs. (This set of
encrypted shares is denoted $C$ in the \textit{Contribute} API.)  The
data source then provides each share to its respective SDP. Secret
sharing ensures data confidentiality, while subsequently encrypting
the shares prevents data misuse: Only correctly attested TEEs (i.e.,
attested to implement the specific $[Q, D, t_e]$ and thus provisioned
with the corresponding decryption keys) can decrypt the shares.

\paragraph{API call ${\bf Query}(q,d,C_1,\ldots,C_K)$}:
To execute a query $q$ over encrypted shares $C_1,\ldots,C_k$ of a
query class $[Q, D, \texp]$, an analyst $d$ sends the set of SDPs of
that class their shares from $C_1, \ldots, C_k$ with the query
$q$. The SDPs independently authenticate $d$, and check that
$q \in Q$, $d \in D$ and current time $< \texp$. Then they decrypt
their shares locally and run MPC to verify the shares, to
reconstruct the plaintexts $m_1,\ldots,m_k$ from the shares, and to
compute $q(m_1,\ldots,m_k)$, which is revealed only to the analyst.

\section{{\sys} scalable analytics}
\label{sec:covault_design}

A key property of the SDP described in the previous section is
that the MPC parties can execute securely within
the same datacenter.  In this section, we show how \sys\ exploits this
colocation and uses the enormous network and compute resources of a
datacenter to execute many subqueries in parallel, thereby scaling out
query processing.  Additionally, we describe how \sys\ manages its database
and efficiently supports oblivious random data accesses.
\iftr
(We defer the description of data ingress
processing to \cref{sec:ingress}.)
\else
\fi

\subsection{Executing queries at scale}
\label{sec:qp-na}
\label{subsec:fga_queries}
\label{subsec:map-reduce}
\label{sec:query-proc}

The unit of querying in \sys\ is a standard SQL filter-groupby-aggregate (FGA)
query of the form:\\
{\small
\texttt{SELECT agg([DISTINCT] col1), \ldots\\
  FROM T WHERE condition GROUP BY col2}}\\
where \texttt{agg} is an aggregation operator like \textsf{SUM} or
\textsf{COUNT}. The query can be executed in three steps: (i) select
(filter) the rows of table \texttt{T} that satisfy
\texttt{condition}, (ii) group the selected rows by \texttt{col2},
and (iii) compute the required \texttt{aggregate} for each group.

\paragraph{Problem \#1: Horizontal scaling.}
The colocation of MPC parties allows \sys\ to leverage
  core parallelism and the aggregate bandwidth available in a
  datacenter to scale-out query processing. For this,
\sys\ draws inspiration from
MapReduce~\cite{dean2004mapreduce}. \sys\ converts FGA queries into a
set of map and reduce tasks, and distributes them across the datacenter
for execution by sets of SDPs using MPC.  Since FGA queries tend to be
highly data-parallel, most map and reduce tasks are data-independent,
and can be executed in parallel using the cores, servers, and
  bandwidth available in a datacenter.

\sys's map and reduce tasks for steps (i)--(iii) of a FGA query are
built using basic oblivious algorithms: bitonic
sort~\cite{Batcher1968}, bitonic merge of sorted
lists~\cite{Batcher1968}, and a butterfly circuit for list compaction
(\S 3 in~\cite{Goodrich2011}) that moves marked records to the end
of a list.  Each of these algorithms is slower than the fastest
non-oblivious algorithm for the same task by a factor of $O(\log(N))$,
where $N$ is the input size.

Oblivious query execution additionally requires that the volume of
intermediate results exchanged among map and reduce tasks not reveal
private information. To this end, the data exchanged between a pair of tasks
is padded to its maximum possible size, which is typically a function of the maximum number of groups in the query's output. Compaction is crucial
  to minimizing this size and thus to minimizing the bandwidth
  needed in the reduce phase.

\paragraph{Problem \#2: Integrity of intermediate results.}
When intermediate data is passed between map-reduce or
  reduce-reduce tasks, integrity of the data must be
  enforced. While we could reuse the authenticated secret-sharing
  scheme to enforce input integrity, MAC computation in MPC is
  expensive. Depending on the MPC protocol used, it may be possible
  instead to exchange data using an in-protocol representation. The
  \sys\ prototype uses this optimization as described in
  \cref{sec:full_construction_design}.

\subsection{{\sys} Database}
\label{subsec:database}

So far, we have not discussed how the encrypted data shares, denoted
$C$ in \cref{sec:fe-construction},
are managed in our prototype. Once a data source has obtained its
encrypted data shares $C$ by locally executing the \emph{Contribute}
API, it forwards each share to the corresponding SDP, which stores the
share in a per-party untrusted database (DB). In the \emph{Query} API
call, the encrypted data shares $C_1,\ldots,C_k$ that the query runs
on are retrieved directly from these DBs; they are not provided by the
querying analyst.
Shares stored in the untrusted DBs are encrypted and MAC'ed with keys
known only inside the SDPs' TEEs.  The organization of these databases
and whether any processing on source data is performed by the SDPs
during data ingress depend on the DB schema and the queries in the
query class. We sketch an example scenario in
\cref{subsec:eva_scenario}. In general, data shares may be stored in
tables with both row- and column-level MACs to enable efficient
integrity checks when the data is read during query processing.

Next, we discuss database access challenges and solutions.

\paragraph{Problem \#3: Efficient oblivious DB random access.}
As explained above, the SDPs read the
inputs $C_1,\ldots, C_k$ from their
DBs to execute queries. For the most part, DB tables are read sequentially in
their entirety. However, some queries need random access to specific
rows for efficiency. The pattern of such accesses can leak secrets
if the locations of the accessed rows depend on secrets read earlier
from the DB (secret-dependent accesses).

In principle, one could implement ORAM~\cite{10.1145/233551.233553}
within MPC to solve this problem; however, even state-of-the-art
  implementations
  ~\cite{10.1145/3133956.3133967,10.1007/978-3-031-48615-9_16,10.1145/3477132.3483562,Mishra2018OblixAE}
  either do not match \sys's requirements and threat model or
are orders of magnitude slower than our solution, described below.  The
properties of private information retrieval (PIR) are closer to our
requirements, but PIR still has substantial
overhead~\cite{DBLP:journals/iacr/Corrigan-GibbsK19a,jofc-2001-14188,DBLP:conf/pet/OlumofinG10,DBLP:conf/nsdi/WangYGVZ17}. Another
option is to randomly permute secret-shared
tables~\cite{sharedshuffle,DBLP:conf/asiaccs/HollandOW22,DBLP:conf/isw/LaurWZ11},
but existing techniques either have substantial overhead or assume semi-honest adversaries.

\paragraph{Oblivious data retrieval (ODR).} 
\label{para:random-shuffling}
{\sys} instead relies on a custom malicious-secure, oblivious data
retrieval (ODR) scheme, which achieves constant-time lookup at the
expense of off-line work to randomly shuffle tables. (DB accesses
where the accessed locations are independent of secret data need not
use ODR.)
Our ODR scheme combines preprocessing inside MPC with pseudorandom
table shuffling {\em outside MPC} for efficiency.  The scheme works as
follows.  For simplicity, we describe the protocol for the case
  of $n=2$ (2PC). However, the scheme generalizes to any number
  of parties.

\smallskip
\noindent \emph{Preprocessing.} During data ingress, which runs in
MPC, we preprocess every table that requires ODR access. We
encrypt-then-MAC (EtM) each row, and separately EtM the row's primary
key. The secrets used to generate the EtMs can be recovered in MPC only.

\smallskip
\noindent \emph{Shuffling.}  A preprocessed table is shuffled by a
pair of processes, $P_1$ and $P_2$, of different parties. This
shuffling is done \emph{outside MPC} for efficiency.  First, $P_1$
locally shuffles the table by \emph{obliviously sorting} rows,
ordering them by a keyed cryptographic hash over the primary key
EtMs. Oblivious sorting also creates a second layer of encryption over
every row. The keys used for hashing and encryption are freshly chosen
by $P_1$. Next, $P_2$ re-shuffles the already shuffled table, by
re-sorting the table along a keyed hash over $P_1$'s hashes using a
fresh key. This doubly-shuffled table is stored in the DB indexed by
$P_2$'s hashes.

\smallskip
\noindent \emph{Row lookup.} $P_1$ and $P_2$ input the secret keys
they picked during shuffling into 2PC. To access a row with a given
primary key, 2PC computes the position of the row in the shuffled table in constant time by applying the EtM and the two keyed hashes to the primary key.
The position is revealed to both parties, which
fetch the row from the shuffled table. Back in 2PC, the MAC of the
row's EtM is checked, the row is decrypted and the primary key stored
within the row is compared to the lookup key for equality.

\smallskip
\noindent \emph{Properties.} The ODR scheme obfuscates the dependence
of DB row locations on primary keys, and protects the integrity and
confidentiality of row contents from a malicious party. First,
recovering the order of rows in the doubly-shuffled table requires
keys used for shuffling by both parties, which only 2PC has. Second, neither party learns
any row individually since all data is encrypted by preprocessing
using secret-shared keys. Third, any attempt to tamper with a row's
data or swap rows is detected by the checks on the fetched row.

Unlike general, less efficient PIR techniques, our ODR scheme does not
hide whether two lookups to the same shuffled table access the same
row.  To avoid frequency attacks, \sys\ prevents such
  accesses. First, it prefers query plans that access each item only
  once (such plans are preferred for efficiency reasons as well). If such
  a plan does not exist, \sys\ caches already fetched items inside MPC
  in order to eliminate multiple database accesses to the same item.
If this is not possible, a different shuffle is used for each access.
Table shuffles can be precomputed, so that freshly shuffled
tables are readily available to queries.  (Preprocessing happens once
per table.)

\paragraph{Optimization: Public index.}
\label{subsec:optimizations}
We can avoid the ODR overhead for queries that perform random access
only on {\em public} attributes by creating \textit{public
  indexes}. Public indexes can also speed up queries that
\textit{join} data on a public attribute (an example of this join
optimization is in
  \iftr
  \cref{subsec:epidemics_ip}).
  \else
  the TR~\cite{covault-tr}, \S D.2).
  \fi
Otherwise, data-oblivious joins can be very
expensive~\cite{Zheng2017}, even more so in MPC.
In \cref{subsec:eva_scenario}, we show illustrative queries that exploit public
indexes.

\paragraph{Problem \#4: Detecting database roll-back.}
A malicious platform operator or a corrupt party may roll back the
database to an earlier version. In general, rollback can be detected by
known techniques~\cite{rote-ussec17,parno2011memoir,nimble},
such as a secure, persistent, monotonic counter within a TEE or TPM,
or a distributed rollback detector.
If a distributed rollback detector is used in \sys, its replicas should be
in different datacenters to prevent a correlated rollback. This does
not impact query performance, because the detector service is
consulted only during reboot of the MPC TEEs and database
modifications. 

In specific cases, such as the epidemic analytics scenario of
\cref{subsec:eva_scenario}, we can instead exploit the fact that the
database is append-only and continuously indexed.  In this case,
rollbacks other than truncation are trivially detectable. To detect
truncation, \sys\ additionally warns the querier
whenever records within the index range specified in the query are
missing.

\paragraph{Problem \#5: Data expiration.}
\sys~must not return a query result if its execution finishes after
the query class's expiration time.  To tolerate a malicious platform
operator who may manipulate the passage of time experienced by a TEE,
the TEEs rely on external third-party timestamping services.  Once a
query execution finishes, a TEE requests a timestamping service to
sign the hash of the query and its result. The TEEs deliver the result
to the analyst only if the timestamp precedes the expiration time.
Multiple timestamping services are contacted to distribute trust.

\subsection{Query expressiveness and limitations}

In principle, \sys\ can run all FGA queries.
Most FGA queries also execute efficiently, but there are exceptions.
First, if a tight upper bound on the number of groups in a reducer's output cannot be determined, then efficiency is lost due to excessive padding of the reducer's output data.
Second, a query loses efficiency if the maximum number of times the same record is accessed via ODR cannot be bounded tightly.

Finally, joins are inherently expensive in MPC. Prior systems implementing
joins in MPC either assume a weaker threat model
(e.g.,~\cite{Volgushev2019}) or they are limited to modest datasets
(e.g.,~\cite{Senate21}). We perform a simple join on small datasets
(user uploads from a certain space-time region) during ingress
processing 
\iftr
(\cref{subsec:epidemics_ip}).
\else
(\S D in the TR~\cite{covault-tr}).
\fi In general, it is
best to do joins during ingress processing if needed and produce
materialised views, which can be queried subsequently without joins.

\section{\sys\ prototype implementation}
\label{sec:full_construction_design}

Next, we detail a specific implementation of the {\sys} core API
(\cref{fig:fe-api}) for $n = 2$ parties, of which $t = 1$ parties may be
malicious. This prototype implementation uses garbled circuits
(GCs) and the DualEx protocol~\cite{dualex} for 2PC.  We
emphasize that the constructions described here generalize to any
number of parties and any MPC protocol (where the adversarial
  model of \sys\ depends on the underlying MPC protocol's threat model, be it semi-honest, malicious with a 1-bit leak, or fully malicious).

\paragraph{Choice of MPC protocol and number of parties.}
The choice of MPC protocol and number of parties in our prototype is
pragmatic. Using two parties aligns with the current availability of 
two different hardware TEE types in public Clouds (Intel TDX and ADM
SEV-SNP). Even with only two parties, an adversary must compromise 
both TEE types, which presents a formidable challenge. The choice
of GCs with DualEx provides a very strong threat model
(active compromise of one TEE with a 1-bit leak) at a resource cost
that is only slightly more than twice that of a semi-honest 2PC
protocol, and a similar runtime.

Choosing instead a semi-honest 2PC protocol is possible but would
imply a total loss of security if one TEE is actively
compromised.  Choosing a fully malicious 2PC protocol is also possible
but would introduce substantially higher latency for each map and
reduce task. However, \sys\ would still scale out up to the available
resources (cores and network bandwidth) in the datacenter.  

\paragraph{DualEx.}
DualEx~\cite{dualex} is an actively secure 2PC protocol that is nearly
as fast as passively-secure protocols and needs only twice as many
cores, but can leak 1 bit of information in the worst case. DualEx
runs two instances of a standard passively-secure 2PC protocol
concurrently on separate core pairs, with the roles of the two
parties, conventionally called \emph{generator} and \emph{evaluator},
reversed. The results of the two runs are compared for equality using
any actively-secure 2PC protocol. If the results mismatch, query
  execution is suspended pending a manual inspection of the system.

In principle, the leaked bit can be the result of any boolean
  function of the attacker's choosing, evaluated on the data consumed
  by a legitimate, approved query.  However, an adversary faces many
  constraints for a practical attack, and recent
  work~\cite{10.1007/978-3-662-46494-6_11} has shown how to
  efficiently limit both the type of boolean function, and the
  probability of an actual leak, enabling a continuum of practical
  security-efficiency tradeoffs. If an application cannot tolerate the
  possibility of a 1-bit leak, \sys\ can be configured with a full
  malicious-secure MPC protocol; we estimate the cost of this choice
  in \cref{sec:evaluation}.  

\paragraph{Adapting DualEx to \sys's needs.}
\label{subsec:dualex}
Like most 2PC protocols, DualEx reveals the computation's result to
both parties. However, we require a protocol that reveals the result
only to the analyst, not to the two parties. To address this
challenge, we combine DualEx with an authenticated secret-sharing
scheme (which we call MtS, see below): We run DualEx to first compute
the result $r = q(m_1,\ldots,m_k)$ and then secret share $r$ two ways
using MtS, producing two authenticated shares $r_1$ and $r_2$.  DualEx
outputs one share to each party; each party forwards its share to the
analyst, who verifies the shares and reconstructs the query result.

\paragraph{Authenticated secret sharing.}
\label{subsec:macs}
The authenticated secret-sharing scheme used in our prototype
  relies on a MAC followed by secret sharing, so we call it
\emph{MAC-then-Share} (MtS).\footnote{{\sys}'s design is not tied to
this particular authenticated secret-sharing scheme. We could also
have used other
schemes~\cite{damgaard2012multiparty,eskandarian2021clarion}.}

The scheme assumes a MAC function $M$ that provides key- and
message-non-malleability as well as message privacy.  To split data $m$ into
two shares, the data holder generates a random key $k$, computes a tag
$t \leftarrow \mathsf{M}_k(m)$, then generates two random strings
$r_k, r_m$, and uses (standard) xor-secret-sharing to generate shares
$k_1, k_2$ of the key $k$ and shares $m_1, m_2$ of the data $m$:
\[ k_1 \leftarrow k \oplus r_k \mbox{,~~~~} k_2
\leftarrow r_k, \mbox{ ~~~~~~~~~~~~ } m_1 \leftarrow m \oplus r_m
\mbox{,~~~~} m_2 \leftarrow r_m\] The two MtS authenticated shares of
the data are $(m_1, k_1, t)$ and $(m_2, k_2, t)$. Each share looks
random, but parties possessing the shares can verify them \emph{jointly} by
checking that $\mathsf{M}_{k_1 \oplus k_2} (m_1 \oplus m_2) = t$. This
verification fails if either share has been modified.

Our prototype instantiates MtS
with $M = KMAC256$~\cite{kmac-nist}. We also tested field-based linear
information-theoretic MACs~\cite{CarterWegman1979,Krawczyk1994} for
$M$, but we empirically found them to be slower than KMAC256 in our
GC-based setting.

\paragraph{In-circuit representation of intermediate results.}
When passing intermediate results between successive GCs (e.g.,
between a map and a reduce circuit), the \sys\ prototype treats
successive GCs as a single, unified circuit and passes intermediate
data in its native in-circuit representation, specifically in
garbled form.  While this choice increases the size of data
transmitted or stored by a factor of 256x in our prototype (128x for
the garbled encoding of bits and 2x for DualEx), we empirically
found it faster than creating and verifying MACs, which
  would otherwise be necessary to ensure integrity when transmitting
  intermediate results.  Our design minimizes in-2PC MAC usage
 in the common case: we verify one MAC for every batch
of data uploaded by a data source, create per-column MACs when storing
data in the database after data ingress, verify these MACs when the
data is read for a query, and create one MAC for every query's
output.  Additional MAC and hash operations are needed for 
ODR-based database accesses, as detailed in \cref{subsec:database}.

\paragraph{Security analysis and proofs.}
\label{subsec:security-analysis}
Our assumption (\cref{subsec:threat-model}) that at least $t = 1$ TEE
 implementation remains uncompromised and that data sources and analysts
 interact only with attested SDPs together imply that at least one SDP
 per query class is fully honest. As a result, all GCs run securely. To
 establish end-to-end security, it remains to show that our MtS
 scheme and the adapted DualEx protocol guarantee confidentiality and
 integrity, which we do in
\iftr
  \cref{app:proofs}.
\else
  the TR~\cite{covault-tr}, \S F. 
\fi

\section{Evaluation}
\label{sec:evaluation}

Next, we present an experimental evaluation of {\sys} to answer the
following high-level questions: What is the cost of basic query
primitives and the shuffling required for ODR? How does query latency
scale with the number of cores for a realistic set of epidemic
analytics queries at scale? How does {\sys}'s performance and scaling
compare to related work with a comparable threat model?

\subsection{Prototype and experimental setup}
\label{subsec:proto}

We implement {\sys} on top of the EMP-toolkit~\cite{emp-toolkit},
a state-of-the-art framework for garbled circuits in C/C++.
First, we implement DualEx\footnote{The implementation of the original
protocol~\cite{dualex} is not publicly available, but reportedly based
on a different Java framework~\cite{fastgc}.} and our extension using
emp-sh2pc for the semi-honest executions, and emp-ag2pc for the
  actively-secure result equality check.
  We use emp-ag2pc also for the
   experiments
  with AGMPC (AG2PC refers to AGMPC with 2 parties).
  Then, we implement circuits for SHA3-256 and AES-128-CTR (which we
  integrate in emp-tool), and the MtS of~\cref{subsec:macs} (based on
  KMAC-256, on top of a pre-existing circuit for Keccak).  Using
  these building blocks, we implement our ODR scheme
  (\cref{para:random-shuffling}), map-reduce primitives (e.g., filter,
  group-by, aggregate), microbenchmarks, and queries.  We use Redis
  (6.0.16, non-persistent mode) as the DB, but
during data-dependent query execution, shuffled views are held in
memory on SDPs rather than in Redis to enable efficient re-shuffling.

We deploy the prototype on Google Cloud Compute Engine (GCE),
  which offers VMs with both AMD SEV-SNP and Intel
  TDX TEE hardware support, the latter available through a restricted private
  preview. Unless otherwise stated, we use the following instance types:
  n2d-standard-8 with Confidential Computing (AMD EPYC 7B13 @ 2.44GHz,
  8vCPU, 4 core, 32GB RAM) and c3-standard-8 under GCE's TDX Private
  Preview (Intel Sapphire Rapids @ 2.7 GHz, 8vCPU, 4 core, 32GB RAM).
  All machines run Ubuntu 22.04 LTS, have a 100 GB
  balanced persistent disk, are located in the same us-central1-a zone, 
  and use Google gVNIC (up to 16Gbps).

\subsection{Microbenchmarks}
\label{sec:microbenchmarks}
We first evaluate our choice of MPC protocol by comparing the
  performance of DualEx with both a semi-honest and a fully
  malicious protocol.  Then, we report the costs of basic oblivious
algorithms (\cref{subsec:fga_queries}), mappers and reducers for a
generic FGA query (\cref{subsec:map-reduce}), and shuffling
(\cref{para:random-shuffling}). These primitive costs provide a basis for
estimating the total cost of arbitrary FGA queries, beyond the specific
epidemic analytics queries we evaluate in
\cref{sec:eval-querylatency}.

\paragraph{GC vs. DualEx vs. AGMPC.}
 To illustrate the performance tradeoff in the choice of a 2PC
 protocol, we compare DualEx
 to two alternatives: (i)
  a semi-honest 2PC protocol using GCs, which also serves as a building block for
  DualEx; and (ii) AGMPC, a state-of-the-art fully malicious-secure
  boolean MPC protocol~\cite{agmpc}, used in systems like Senate~\cite{Senate21}.
  While AGMPC generally assumes $t=n-1$, we compare against 
  AGMPC with 2 parties (AG2PC), where $t=1$ as in DualEx.

We sort (the most expensive query primitive) 10,000 32-bit inputs using
two parties. We evaluate AG2PC, semi-honest
GCs (a single execution), and DualEx (two symmetric,
sequential runs plus the equality check), running each protocol
on an Intel TDX and a AMD-SEV 
VM corresponding to the two parties, with one core per party.  In
addition, we run DualEx with two cores per party, thereby exploiting
the natural concurrency in this protocol by executing the two
semi-honest GC concurrent runs in parallel. The
results are shown in \cref{fig:break}.

With the same number of cores (one per party), DualEx runs
  twice as long as a semi-honest GC execution, as
  expected, and DualEx is 11.36x faster than the fully
  malicious-secure AG2PC. When given extra cores to exploit its
  inherent parallelism, DualEx is only 8\% slower than the
  semi-honest GCs.  These results suggest that DualEx is a compelling choice in
  scenarios where a 1-bit leak is acceptable.

\begin{figure}[h]
	\centering
	\includegraphics[width=0.45\textwidth]{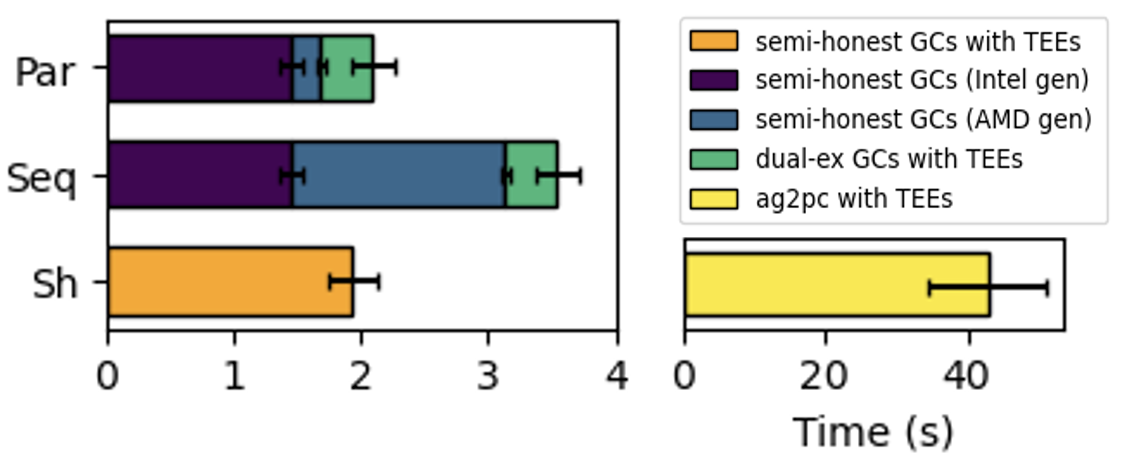}
	\caption{Cost breakdown. \textit{Par} and \textit{Seq} indicate parallel and
          sequential execution of the two GC computations in
          DualEx. \textit{Sh} is the semi-honest 2PC execution in TEEs.}
	\label{fig:break}
\end{figure}

\begin{figure}
	\centering
	\includegraphics[width=\columnwidth]{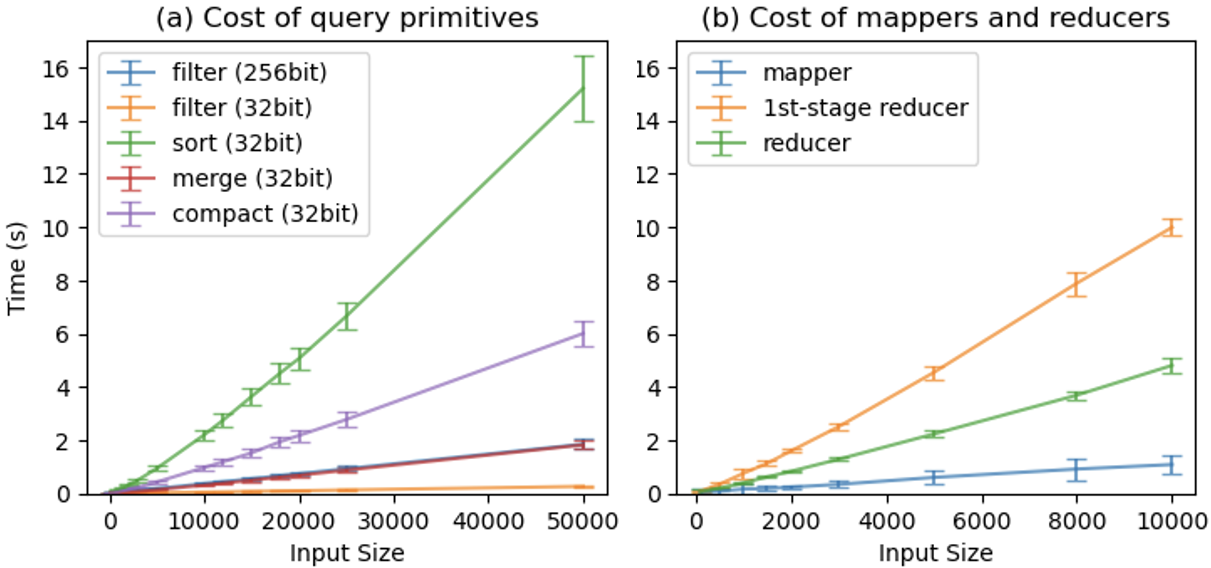}
	\caption{2PC processing time vs. input size for basic oblivious
          algorithms and mappers/reducers.}
	\label{fig:eval-triple}
	\label{fig:eval-primitives}
	\label{fig:eval-map-reduce}
\end{figure}

\paragraph{Cost breakdown.}
\Cref{fig:break} shows the latency contributions of the
individual components in DualEx.  The purple and blue parts show the
contributions of the two parallel GC executions required for DualEx,
one where the Intel CPU plays the role of the generator, and one with
the AMD CPU as the generator. (The two latencies differ slightly
because the load on the generator is higher, and the Intel CPU is slightly
faster; the slower execution determines the total latency.)  The green
part of the bars indicates the additional latency incurred by executing
the two parties in TEEs (SEV-SNP and TDX), along with the
result comparison overhead, which is negligible.

Differentially private queries require the addition of appropriate
noise to the final result of a query. The cost of adding such noise is
negligible, on the order of ms.

\paragraph{Query primitives.}
\Cref{fig:eval-primitives}a shows the execution time for the basic
oblivious query primitives from
\cref{subsec:fga_queries}---linear scans on 32-bit and 256-bit
records, sorting, merging two sorted lists, and compaction on 32-bit
records---as a function of the list length (input size). Each reported
number is an average of 5000 measurements, with std. dev. shown as
error bars.
The trends are as expected: The cost of linear scans grows linearly with
the input size, sorting, sorted merge, and compaction exhibit slightly 
super-linear growth. Sorting is
significantly more expensive than compaction, which is why our reduce
trees sort only in the first stage and merge-compact in subsequent
stages.

\paragraph{Oblivious MapReduce.} 
\Cref{fig:eval-map-reduce}b shows the average execution time of
typical mappers and reducers on a single core pair with DualEx in
TEEs.  The specific operations are those of query q2 in
\cref{fig:queries-eval}.

  The complexities of a mapper, 1st-stage reducer, and
  subsequent-stage reducer are $O(c)$, $O(c
  (\log(c))^2)$, and $O(d \log(d))$, respectively, where $c$ is the input size and $d$
  is the maximum width of a reducer's result. The 1st-stage reducers
  are more expensive due to the additional sort operation.
The results align with expectations: Map costs are linear in the input size,
while reduce costs are slightly super-linear, with the 1st-stage
reduction being the most expensive.

\paragraph{Shuffling.}
Each shuffle requires two sequential oblivious sorts in TEEs of
different types, outside 2PC.  Shuffling a view of 600M records
takes 58min on a single core pair.%
\footnote{600M records exceed the 32GB RAM 
	in our standard VM configuration. So, for the experiment with 600M records, 
	we increased the RAM size to 88GB in the SNP VM, and
	we used a c3-standard-22 VM instance with 88GB RAM as the TDX VM.}
Shuffling one-tenth the number, 60M records, takes 4.6min.  The variances are
negligible. The scaling is super-linear since oblivious sort runs in $O(N (\log(N))^2)$ time
even outside 2PC.

A shuffle is used only once for a query that requires ODR, but shuffles
can be precomputed in parallel. Since a shuffle sorts on
one core at a time, a single core pair can
produce two different shuffles of 600M records in under 1h---a conservative upper
bound on the pairwise encounters generated by a country of 80M people
in 1h in our epidemic analytics scenario
(\cref{subsec:eva_scenario}). Thus, for a country of this size, we
can prepare shuffles for $q$ queries using $q/2$ core pairs
continuously.

\paragraph{Bandwidth.}
Garbled circuits are streamed from the generator to the evaluator, 
along with the garbled encoding of the evaluator's inputs.
During a series of sort and linear scan operations, the average
bandwidth from the generator to the evaluator, measured using
NetHogs~\cite{nethogs}, is $\sim$2.8Gbps.  The bandwidth from the
evaluator to the generator is negligible in comparison. With DualEx,
the average bandwidth is 2.8Gbps in each direction.
We use 4 active cores per machine for a bandwidth of 5.6Gbps
in each direction, which the GCE gVNIC can support. 
For comparison, the average bandwidth requirement when using AG2PC
	 is just under 2Gbps in each direction.

\subsection{End-to-end scenario: Epidemic analytics}
\label{subsec:eva_scenario}

We evaluate \sys\ at scale using epidemic analytics as an
example scenario. We use DualEx as the underlying 2PC protocol
throughout, unless stated otherwise.

Since \sys\ uses oblivious algorithms, the actual data values
are irrelevant for performance: what matters is the database schema,
the number of records, and the query structure.  Hence, our evaluation
generalizes to any scenario with similar data sizes and
sequence of FGA operations.

For our epidemic analytics scenario, we use {\em synthetic} location
and Bluetooth Low Energy (BLE) radio encounter data---which would
normally be collected via smartphones---as detailed in
\iftr 
\cref{sec:epidemics}.
\else
the TR~\cite{covault-tr} (\S D). 
\fi
Ingress processing produces two materialized views, $T_E$ and $T_P$,
whose schemas are shown in \cref{fig:tab}. Each view includes a
\emph{public}, coarse-grained index, marked with \emph{italics}
font. $T_E$ stores pairwise encounters, containing an encounter ID
(eid) and the anonymous IDs of the two devices (did1, did2).
The public index is keyed with a coarse-grained \emph{space-time region}.

The second view, $T_P$, contains encounters \textit{privately} indexed
by individual device IDs and the times of the encounter reports (did1,
time). Each record includes pointers to previous and next
encounters, enabling timeline traversal for a given device. 
The public index on $T_P$ is a coarse-grained
\emph{epoch} ($\sim$1h) during which an encounter occurred.
$T_P$ is accessed via the private index in data-dependent queries; thus,
$T_P$ is shuffled as required by the ODR scheme (\cref{para:random-shuffling}).

\begin{figure}[t]
    \caption{Schema and queries used in epidemic analytics. The selections on the public attributes \texttt{space-time-region} and \texttt{epoch} are done outside 2PC using the public indexes of $T_E$ and $T_P$. \texttt{R} is a set of space-time-regions.}
  \label{fig:queries-eval}

    \label{fig:tab}
  \centering
  \scalebox{0.8}{
      {\small %
   	\setlength\tabcolsep{1.5pt}
   	\begin{tabular}{@{} c | c | c | c | c | c | c | c | @{}}
   	  \cline{2-8}
   	  $T_E$ & \textit{space-time-region} & eid & did1 & did2 & \multicolumn{3}{c|}{\ldots} \\
   	  \hhline{~=======}
   	  $T_P$ & \textit{epoch} & did1, time & did2 & duration & prev & next & \ldots \\
   	  \cline{2-8}
   	\end{tabular}
   	\setlength\tabcolsep{6pt} %
  } }

  \centering
  \scalebox{0.8}{
  \begin{tabular}{@{}|l|@{}}
    \hline
    \begin{tabular}{@{}l@{~}p{0.4\textwidth}@{}}
      \textbf{(q1)} & Histogram of \#encounters, in space-time regions \texttt{R}, of devices in set \texttt{A}
    \end{tabular}\\

    \begin{tabular}{@{}l@{}}
    \small{\texttt{SELECT HISTO(COUNT(*)) FROM $T_E$}}\\
    \small{\texttt{WHERE did1$\,\in\,$A AND space-time-region$\,\in\,$R}}
    \end{tabular}\\
    \hline
    \begin{tabular}{@{}l@{~}p{0.4\textwidth}@{}}
      \textbf{(q2)} & Histogram of \#unique devices met, in space-time regions \texttt{R}, by each device in set \texttt{A} 
    \end{tabular}\\
    \begin{tabular}{@{}l@{}}
      \small \texttt{SELECT HISTO(COUNT(DISTINCT(did2))) FROM $T_E$}\\
      \small \texttt{WHERE did1$\,\in\,$A AND space-time-region$\,\in\,$R}
    \end{tabular}\\
    \hline
    \begin{tabular}{@{}l@{~}p{0.4\textwidth}@{}}
      \textbf{(q3)} & Count \#devices in set \texttt{B} that encountered a device in
      set \texttt{A} in the time interval \texttt{[start,end]}
    \end{tabular}\\
    \begin{tabular}{@{}l@{}}
     \small \texttt{WITH TT AS} \\ 
     \small \texttt{~(SELECT * FROM $T_P$} \\
     \small \texttt{~~WHERE start$\,<\,$epoch$\,<\,$end)} \\ 
      \small \texttt{SELECT COUNT(DISTINCT(did2)) FROM TT}\\
      \small \texttt{WHERE did1$\,\in\,$A AND did2$\,\in\,$B} \\
      \small \texttt{~~~~~~AND start$\,<\,$time$\,<\,$end}
    \end{tabular}\\
    \hline
  \end{tabular}}
\end{figure}

Our evaluation uses the queries q1--q3 in
\cref{fig:queries-eval}, developed in consultation with an
epidemiologist.
These queries can help assess the impact of contact
restrictions, such as the closure of large events, on the frequency of
contacts among people during epidemics.
Query q3 can help determine whether two epidemic outbreaks
(represented by device sets \texttt{A} and \texttt{B}) are
directly connected by an encounter. A related query, discussed in
\iftr \cref{sec:eval-additional-query}, \else the TR~\cite{covault-tr} (\S
D.4),
\fi
extends q3 to indirect encounters between \texttt{A} and \texttt{B}
via a third device.

\begin{figure}
	\centering
    \includegraphics[scale=0.45]{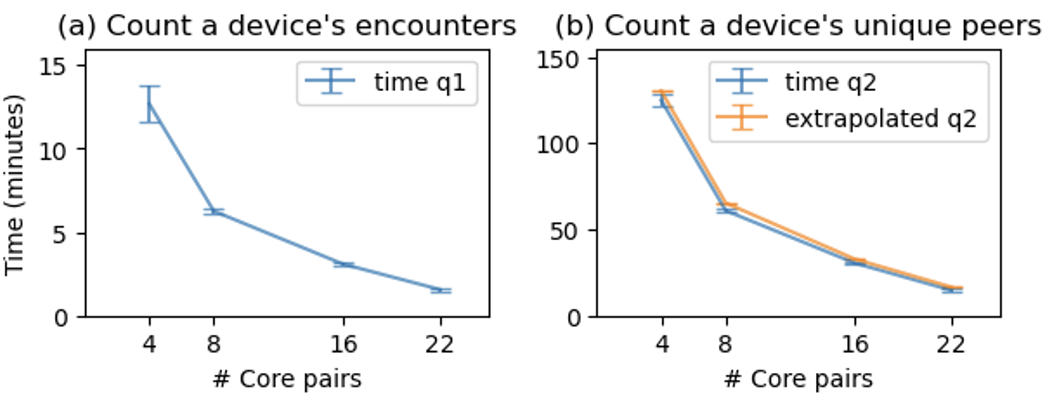}
	\caption{Query latency vs.\ total number of core pairs
            (DualEx) for the FGA queries (q1) and (q2) from
            \cref{fig:queries-eval}.  The measurements are the average
            of 50 runs, with std. dev. shown as error bars.  In (b),
            we additionally show the performance predicted by our
            performance model (described below). }
	\label{fig:eval-query-latency-actual}
	\label{fig:eval-query-latency-actual:count}
	\label{fig:eval-query-latency-actual:set}
\end{figure}

\label{sec:eval-querylatency}

We measured end-to-end query latency for the queries q1--q3
as a function of input size and core count.

\paragraph{Queries q1 and q2.}
Both q1 and q2 run on the view $T_E$. The queries fetch only the
records in the space-time region \texttt{R}, using the public index of
$T_E$, which significantly reduces the amount of data processed
  in MPC and hence the query latency.  In our synthetic
  database, the space-time region \texttt{R} contains 28M encounter
records, representing a conservative upper bound on encounters
generated in a space-time region with 10k data sources reporting 200
encs/day over 14 days.  We process these records in 10k-sized shards,
empirically minimizing latency.  In both queries, we iterate over
devices in the given set \texttt{A}. For each device $a$ in \texttt{A}
we issue a FGA query. For q1, this FGA query counts the encounters of
device $a$ in \texttt{R}.  For q2, the FGA query counts the number of
unique devices that $a$ met in \texttt{R}. These FGA queries can
execute in parallel for all devices in \texttt{A}.

\Cref{fig:eval-query-latency-actual:count}a shows the latency of q1
as a function of available core pairs (currently limited to 22 in
the GCE TDX private preview).  Mappers filter input encounters to
those involving the specific user $a$, while reducers aggregate
encounter counts. The results show that latency scales almost inversely with 
the number of core pairs available, showing near-linear horizontal scaling.

\Cref{fig:eval-query-latency-actual:set}b shows the latency of q2. 
The map phase is unchanged, but reducers must
merge and deduplicate lists of encountered devices.
Again, query latency varies almost inversely with the number of cores,
but is higher than for q1, since 1st-stage reducers perform more work in
q2.

\paragraph{Query q3.}
Query q3 runs on the view $T_P$, and
computes the number of devices in set \texttt{B} that
  directly encountered a device in set \texttt{A} within the time
  interval \texttt{[start,end]}.  A na\"ive 2PC implementation
would require a linear scan of all encounter records in this interval to
find those between two peers from sets \texttt{A} and
\texttt{B}. Based on our earlier estimate, for a table of 80M people
and a 14-day period, this would involve loading and filtering
$165.9$B records \emph{in 2PC}.

Instead, q3 leverages $T_P$ to pre-select relevant epochs in
\texttt{[start,end]} outside 2PC. For each relevant epoch $e$ and
device $b$ in \texttt{B}, the query traverses $b$'s sequence of
encounters by dereferencing \texttt{prev} pointers obliviously using
ODR, starting from $b$'s latest encounter.  Dummy lookups are inserted
as needed to hide the actual number of encounters $b$ has in $e$. For
each such encounter, the query checks if $b$'s time falls in
\texttt{[start,end]} and if $b$'s peer is in \texttt{A}, using an
in-circuit Bloom filter initialized with \texttt{A}. If so, this $b$
encountered a device in $\texttt{A}$. The query repeats this traversal for every $b$ in $\texttt{B}$ and all relevant epochs, thus counting the number of devices in $\texttt{B}$ that encountered some device in $\texttt{A}$.

For \texttt{A} and \texttt{B} of size 10 each and a 14-day period,
this approach reduces the number of records processed in 2PC by five
orders of magnitude, down to \textasciitilde$100$K. The total query latency
on 2 core pairs is 2.44 hours, assuming shuffled
views (for ODR) for each of the 336 1h epochs are precomputed.

Preparing two shuffled views of an epoch took
less than one hour on a single core pair
(\cref{sec:microbenchmarks}), and shuffled views can be precomputed
upfront and in parallel whenever idle cores are available.
Given q3's modest latency using 2 core pairs, we did not parallelize
its implementation.  However, q3 is highly data-parallel. For
instance, ``checking whether a given device in \texttt{B}
intersects \texttt{A} during a 1h epoch'' can be treated as a sequential
execution unit, which can be run in parallel on separate cores.
This unit takes $\sim$2.6s in our prototype. The time required to add
up the final counts is insignificant by comparison.  The parallel
sub-computations share one ODR view per epoch,
as they access disjoint sets of records.

\subsubsection{Estimating query latency at scale}
\label{subsubsec:estimating}

  \paragraph{Performance model.}
  So far, we have presented measured results from a modest-sized
  deployment in GCE with up to 22 core pairs. Due to restrictions in the 
  GCE TDX private preview, additional TDX cores are unavailable to us;
  at any rate, larger-scale experiments with thousand of cores
  would also exceed our monetary budget. Instead, we developed a performance model based on detailed
  measurements of individual mappers, reducers, and query primitives
  to be able to extrapolate the performance of large deployments. This
  model accurately predicts our measured query performance for small
  deployments, as shown in
  \Cref{fig:eval-query-latency-actual:set}(b). Details of the model
are in
\iftr
\cref{sec:mr-extrapolation}.
\else the TR~\cite{covault-tr}, \S E.
\fi

\paragraph{Scaling to many cores.}
Using our model, we extrapolate the number of core pairs needed to
achieve a target query latency for a given number of input
records. For example, if the input was 10x larger (280M records),
answering q2's basic query in 10min would need 272 core pairs (e.g.,
34 16-CPU machines).
Given our measured bandwidth for 2PC (\cref{sec:microbenchmarks}),
executing this query with 272 core pairs requires a total
bisection bandwidth of $\sim$200Gbps between pairs of servers
hosting the parties. While a datacenter can easily meet this demand among
servers in the same rack, a geo-distributed MPC system would
require similar bandwidth {\em across} datacenters! This shows that
colocation is a key to scale-out in MPC-based computations.

\paragraph{Scaling to a country.} To perform epidemic analytics for
an entire country with 80M people, our model predicts that continuously
ingesting encounter records (11.85B records/day) requires
1,660 core pairs running continuously (see
  \iftr \cref{sec:eval-ingress}).
  \else the TR~\cite{covault-tr}, \S D.3).
  \fi
Furthermore, running q2 on 14 days of records (165.9B total) within 24h
requires an additional 1,074 core pairs engaged for those 24h.
Looked at differently, a core pair is required for every
$\sim$30,000 citizens at this scale.  While this resource cost is
  substantial, it remains within the capacity of even medium-sized
  datacenters. Moreover, we believe that the monetary cost is justifiable for
  high-value analytics like those relevant for public health.

\paragraph{Performance with fully malicious and semi-honest 2PC.}
All reported query results used DualEx as the MPC
protocol. Here, we estimate the performance impact of using 
AG2PC, a fully malicious-secure protocol, and GCs, a
semi-honest 2PC protocol for comparison. To do so, we
measured the performance of mappers, reducers, and query
primitives under these protocols and used our performance model to
extrapolate query performance at scale.

With AG2PC, executing q2 as described above on 14 days of records 
within 24h requires an estimated 13,000 core pairs, over
12x more than with DualEx. With current state-of-the-art 2PC
protocols, this is the cost of a fully malicious threat model that
does not admit the 1-bit leak. The bandwidth requirement per core pair
remains below 2Gbps on average, peaking below 6Gbps,
well within the capacity of a datacenter even at this
scale (Google reportedly has more than 1Pbps bisection bandwidth~\cite{singh2015jupiter}).
Thus, \sys's colocation of parties allows it to scale out
regardless of the MPC protocol and enables analytics at scale
even in a fully malicious model, provided sufficient datacenter
resources.

For comparison, executing the same q2 query within 24h using a semi-honest GC
protocol is possible with 537 core pairs (half of
the core pairs required for DualEx) since only one of the two parallel
semi-honest GC executions is required.

\vspace{\baselineskip}
\subsection{{\sys} versus existing work}
\label{subsec:senate-eval}

Existing secure analytics platforms with active security fall into two
broad categories. In {\em cooperative analytics} platforms like
Senate~\cite{Senate21}, large data aggregators run
distributed MPC to analyse their collective, distributed data. As a
result, the number of data sources participating in a (sub)\-query
determines the number of MPC parties, and the MPC parties are
geo-distributed.  The overhead of AGMPC, which is used in Senate,
increases super-linearly with the number of parties. Moreover, WAN
bandwidth limits Senate's scalability, restricting its analytics
to small databases.

The Senate prototype is not available; based on published runtime and
network usage results for Senate's $m$-Sort circuit with 16
parties~\cite[Figs.\ 5a, 7b]{Senate21}, we estimate its average
bandwidth during the execution at 4.8Gbps (480GB/800s).  We can use
this result to extrapolate its total average bandwidth requirement if
one tried to perform parallel subqueries on a large dataset. To
execute our query q2 on 28M records, Senate would require 1750
(28M/16k) parallel sorts in the first stage, resulting in a total
average bandwidth of 8.4Tbps!  By contrast, in \sys, the number
of MPC parties is independent of the number of data sources, and the
colocation of parties enables scale-out up to the available cores and
bisection bandwidth within a single datacenter.

{\em Federated analytics (FA)} platforms like
Arboretum~\cite{sosp23-arboretum} analyse data stored on millions of
individual user devices, using a combination of homomorphic
encryption, zero-knowledge proofs, MPC among small committees of user
devices, and a centralized untrusted aggregator. We implemented
Arboretum's \texttt{top1} query in \sys. According to published results,
for $2^{15}$ categories and $10^9$ data sources/inputs,
this query takes $\sim$9h with a 1,000-core aggregator in Arboretum.
Using our performance model, we estimate that 250 \sys\ SDPs (1,000
cores) running AG2PC (to match Arboretum's active security model) would
process this query with a delay in the same order of magnitude.

FA and \sys\ differ in qualitative aspects, making the two approaches
suitable for different scenarios. With FA, query efficiency, results
and repeatability depend on the availability, cooperation, and
resource contributions of individual user devices, as well as WAN
network bandwidth and delays. However, FA does not require trusted,
centralized components.  \sys, in contrast, requires more
substantial datacenter resources, but offers repeatable query results
and additionally supports data-dependent and iterative queries.

\section{Related work}
\label{sec:related}

Even though considerable progress has been made towards scaling secure
analytics platforms and MPC broadly, no system scales horizontally to
billions of records, while meeting the security properties of
\cref{sec:intro}. For example, many existing systems either provide
only passive
security~\cite{Nayak2015,Chowdhury20,190910,stealthdb,DBLP:conf/ccs/OhrimenkoCFGKS15,papadimitriou2016,Bater2017,Volgushev2019,secureml},
or do not fully distribute
trust~\cite{Maheshwari2000,Vingralek2002,Bajaj2011,Arasu2013,Baumann2015,Schuster2015,ryoan-osdi16,fuhry2017,Zheng2017,DBLP:conf/sp/PriebeVC18,oasis-labs1,oasis-labs2,Arasu2013oblivious,oblidb,10.1145/3342195.3387552},
and, thus, satisfy the security properties only partially. Here, we
compare to closely related work, except federated analytics, which we
covered in \cref{subsec:senate-eval}.

\paragraph{Collaborative / cooperative analytics}
refers to MPC-based analytics where the parties are also the data
sources.  SMCQL~\cite{Bater2017} and Conclave~\cite{Volgushev2019} are
only passively secure, while Senate~\cite{Senate21} provides active
security. Cooperative analytics runs early stages of the computation
on subsets of data sources. However, these systems have no support for
data-dependent queries, and only support running queries as
monolithic cryptographic computations, without horizontal scaling to
cores available within each party, which is the primary problem {\sys}
addresses.

\paragraph{Combining MPC and TEEs.}
The encapsulation of MPC parties in TEEs has been considered in
related work for various purposes but not for scaling MPC
horizontally, which is what {\sys} does. Encapsulation in TEEs of
\emph{different vendors} has been used to bootstrap decentralized
trust by enabling stakeholders to audit the MPC
configuration~\cite{dauterman2022reflections} and to recover keys
securely~\cite{DBLP:conf/osdi/ConnellFSDP24}.

Encapsulation of MPC parties of the \emph{same vendor}, or without
explicit consideration of TEE heterogeneity, has been used to obtain
active security from passively-secure protocols~\cite{cryptflow}, to
offload subcomputations to a single party within MPC~\cite{wuhybrid},
to obtain a protocol for fair multi-party transactions in blockchains
without MPC~\cite{ren2023decloak}, and to improve trust in code
verifiers~\cite{toki/multi-prover}. Unlike \sys, these systems cannot
tolerate the failure or compromise of any TEE.

\paragraph{Secret-shared data analytics systems}
aim to decentralize trust. Obscure~\cite{Gupta2019} supports
aggregation queries on a secret-shared dataset outsourced by a set of
owners; however, it does not support big data and nested
queries. Crypt$\epsilon$~\cite{Chowdhury20} and
GraphSC~\cite{Nayak2015} outsource computation to two untrusted
non-colluding servers, but assume only semi-honest adversaries.
Waldo~\cite{9833611} uses secret sharing and honest-majority 3PC 
but focuses on outsourced analytics of \textit{time-series}
data from a \emph{single} source.
Unlike \sys, all three systems rely on strong assumptions about the
absence of correlated attacks.

\paragraph{TEEs without
  MPC}~\cite{Maheshwari2000,Vingralek2002,Bajaj2011,Arasu2013,Baumann2015,Schuster2015,ryoan-osdi16,fuhry2017,DBLP:conf/sp/PriebeVC18,oasis-labs1,oasis-labs2}
protect data at rest and in use by decrypting data only inside a TEE.
This idea has been combined with \emph{oblivious algorithms} to
mitigate side-channel
leaks~\cite{190910,Zheng2017,stealthdb,DBLP:conf/ccs/OhrimenkoCFGKS15,197247,Arasu2013oblivious,oblidb,10.1145/3342195.3387552}.
However, unlike {\sys}, these systems do not distribute trust among
independent parties. They address orthogonal problems, e.g., the
design of efficient oblivious algorithms or TEE-related protocols.

\paragraph{Homomorphic encryption (HE)}
is an alternative to classic MPC. Full HE is prohibitively
expensive~\cite{sok-fhe}.  Partial HE (PHE) restricts query
expressiveness significantly, and works efficiently only in weak
threat models. Seabed~\cite{papadimitriou2016} works in a semi-honest
setting. TimeCrypt~\cite{burkhalter2020timecrypt} and
Zeph~\cite{burkhalter2021zeph} allow data sources to specify access
control preferences (similar to SFC); however, they specifically
target time-series data with a restricted set of operations
(additions but not multiplications), and Zeph operates in a
semi-honest threat model.

\paragraph{Automatic optimization of MPC}
has been considered in Arboretum and several MPC frameworks
\cite{mp-spdz,Bater2017,Volgushev2019,hycc,costco,silph,ezpc}. Heuristics
automatically optimize the computation's data flow and, in some cases,
partition the data flow and select optimal MPC protocols for all
partitions. These optimization techniques are orthogonal to {\sys}'s
design and can be applied to its future implementations.

\paragraph{Encrypted databases} like CryptDB~\cite{popa2011} protect
data at rest only. Early work used weak encryption like deterministic
or order-preserving encryption, later shown to be inadequate for
security~\cite{Grubbs2017,Bindschaedler2018}. Blind
Seer~\cite{DBLP:conf/sp/PappasKVKMCGKB14} uses strong encryption and
2PC to traverse a specialized index, but leaks information about its
search tree traversal. These systems do not protect data in use.

\paragraph{Privacy-preserving data donation systems}
Anonify~\cite{gaballah2024anonify} provides data-source and attribute anonymity for 
donated medical records. This system focuses on anonymization, assumes passive security, cannot handle millions of records, and reveals sampled \textit{anonymized} records.
In contrast, \sys~focuses on secure computation, executes queries, and reveals only query results (with SFC compliance and support for differential privacy).

\section{Conclusion}
\label{sec:conclusion}

{\sys} relies on a set of colocated MPC parties, each executing in a
TEE of a different vendor. Because {\sys} is not impeded by the limits of WAN
communication during query processing, it can scale out and leverage
the resources of a datacenter to support big data analytics over data
from individual sources, with distributed trust and malicious
security.  {\sys} scales out to epidemic analytics queries for a
country of 80M people and billions of records at a cost of a core pair
for every 30,000 people.

\section{Acknowledgments}

We thank Gilles Barthe, Jonathan Katz, Matthew Lentz, Chang Liu, and
Elaine Shi for discussions and contributions to early versions of this
work; Hongjie Liu for advice on relevant epidemiological queries; Xiao
Wang and Chenkai Weng for their help with emp-toolkit; Lorenzo Alvisi,
Natacha Crooks, Aastha Mehta, Lillian Tsai, Vaastav Anand, and several
anonymous reviewers for their feedback on earlier drafts of this
paper. Finally, we are grateful to the anonymous shepherd and
reviewers of Usenix Security 2025 for their insightful feedback and
suggestions.  This work was supported in part by the Max Planck
Society, the US National Science Foundation Graduate Research
Fellowship Program under Grant No. DGE 1840340, the European Research
Council (ERC Synergy imPACT 610150), and the German Science Foundation
(DFG CRC 1223).

\section{Ethics Considerations}

We carefully considered ethical implications throughout the
development of {\sys}. In fact, the motivation of our work is based on
the principle of public welfare without compromising security
(\cref{sec:intro}).

Our design includes explicit user consent (which we call selective
forward consent or SFC). It provides full transparency and control to
data owners in terms of how, by whom and until when their data can be
used. We integrate technical means to enforce SFC with TEEs, and
Appendix A of our TR explains how TEE attestation, which is a building
block, can be made practical using community processes.

Thanks to the data obliviousness of MPC circuits, we did not need any
real data to obtain accurate performance measurements. All our
experiments are based on synthetic data that we generated ourselves.

\section{Open Science}
This paper is accompanied by an open-source artifact, which 
 includes the code for the experiments in \cref{sec:evaluation}, as well as
 the scripts to execute the code and produce results.
The prototype uses synthetic datasets that are generated by the code.
The results reported in the paper were obtained by running the scripts
in Google Cloud Compute Engine following the setup described in 
\cref{subsec:proto}. 
The artifact is available at \url{https://doi.org/10.5281/zenodo.14736567}.

\bibliographystyle{plain}
\bibliography{main_abbr}

\iftr %
\appendix
\section{Community approval process}
\label{sec:community-approval}

As mentioned in~\cref{subsec:threat-model}, \sys\ relies on a
community approval process to \tr{establish} trust in the
system. Specifically, \tr{this process serves} two purposes: (i) to
\tr{attests} the PSs, which \tr{in turn attest and provision} DPs;
\tr{and} (ii) to \tr{assure} that the query classes to which data
sources contribute their data \tr{correctly perform the operations
  defined in their specifications}.  In both cases, data sources can
rely on experts they trust \tr{to attest} the PSs and review the
implementation of PSs and DPs.

Any interested community member \tr{(a witness)} can inspect the source code of each
system component, verify their measurement hashes, remotely attest the
PSs, and publish a signed statement of their \tr{assessment}.  To do so, a
witness performs the following \tr{steps}:
\begin{enumerate}
\item Inspect the source code of the PSs and the build chain used to
  compile all system components, \tr{in order to ensure} they correctly implement
  the system's specification. Verify that the measurement hash \tr{of the PSs
  matches the one}
  recorded in the system configuration, and
  remotely attest the PSs accordingly.
\item Obtain the current system configuration from the PSs.
\item Review the specification and source code of each query in a given
class, and verify \tr{that} the measurement hash of the class's DPs 
\tr{matches the one} recorded in the \tr{system} configuration.
\item Publish a signed statement of their assessment.
\end{enumerate}

Attackers could \tr{attempt to} remove or suppress witness statements,
and \tr{malicious} witnesses could post false \tr{attestation} statements 
\tr{or reviews of PSs and DPs (query classes)}. However, this 
\tr{would at most result in a} denial-of-service (DoS) \tr{attack}, 
\tr{provided that data sources remain vigilant and disregard reviews from 
any witness} they do not fully trust.
Furthermore, the state of the PSs could be rolled
back by deleting or replacing \tr{their} sealed states
with an earlier version. However, this \tr{too would only lead} to DoS, as it
would \tr{render} inaccessible some recently 
defined query classes and their \tr{associated} data.

\section{Scalable analytics queries}
\label{app:details-map-reduce}

The unit of querying in \sys\ is a standard SQL
filter-groupby-aggregate (FGA) query. In general, an analyst may
issue \emph{data-dependent} FGA queries, where database accesses
depend on private database values.  \tr{Below}, we first describe
how \sys\ executes basic FGA queries, and then how it handles
data-dependent queries.

\paragraph{FGA queries.}
\tr{As outlined in \cref{subsec:fga_queries}}, a FGA query 
\tr{follows the general structure}:\\
{\small
	\texttt{SELECT agg([DISTINCT] col1), \ldots\\
		FROM T WHERE condition GROUP BY col2}}\\

where \texttt{agg} is an aggregation operator like \textsf{SUM} or
\textsf{COUNT}. The query can be executed in three steps: (i) select
(filter) the rows of table \texttt{T} that satisfy
\texttt{condition}, (ii) group the selected rows by \texttt{col2},
and (iii) compute the required \texttt{aggregate} for each group.

A straightforward \tr{implementation of} a FGA query \tr{would construct}
a \emph{single} garbled circuit that takes as input the two shares of
the entire table \texttt{T} and \tr{performs} steps (i)---(iii). However,
this approach \tr{has significant drawbacks. First, it does not leverage}
\emph{core parallelism}, \tr{which is essential} for
reducing query latency \tr{at scale}. \tr{Second}, the size of the
circuit \tr{increases} super-linearly with the size of \texttt{T} and
the circuit may become too large to fit into the main memory of a
single machine. To \tr{address these challenges}, {\sys} implements
FGA queries using MapReduce~\cite{dean2004mapreduce}, \tr{which
enables parallelization and scalability.} We first explain \tr{the
general MapReduce approach to execute FGA queries before detailing
how \sys\ applies it in a 2PC setting}.

\paragraph{Background: FGA queries with MapReduce.}
Suppose we have $m$ available cores. The records of table \texttt{T}
are evenly \tr{distributed} across these $m$ cores.

\noindent\textit{Step (i): \tr{Filtering (Map phase).}} 
Each core splits its \tr{assigned} records into \tr{smaller, }more manageable
\emph{chunks} and applies a \emph{map} operation to each chunk. This
operation scans the chunk linearly, filtering out only the records that
satisfy the \textsf{WHERE} \texttt{condition}.

\noindent\textit{Steps (ii) and (iii): \tr{Grouping and Aggregation (Reduce Phase).}} 
These steps are executed through a tree-shaped
\emph{reduce} phase: 
\begin{itemize}
\item First-stage reducers: This phase uses half \tr{as many}
reducers as mappers. Each 1st-stage reducer
receives filtered records from two mappers, sorts them
by the grouping column \texttt{col2}, and then scans them
\tr{linearly} to compute an aggregate for each \tr{unique}
\texttt{col2} value. The output is a sorted list of \texttt{col2}
values with \tr{their} corresponding aggregates.
\item Subsequent-stage reducers: Every subsequent stage of reduce
uses half the number of reducers of the previous stage.  Each reducer
merges the sorted lists \tr{produced by two reducers from the previous
stage, combining} their corresponding aggregates and outputting
another sorted list.
\item \tr{Final-stage reducer: The last remaining reducer
consolidates all intermediate results into a single sorted 
list of \texttt{col2} values with their final aggregates.}
\end{itemize}

\paragraph{\sys: FGA queries with MapReduce in 2PC.}
\sys\ executes FGA queries by implementing each mapper and reducer
as a \emph{separate} garbled circuit (GC), \tr{which is then
evaluated} in a 2PC setting. \tr{This design allows \sys\ to leverage
the scalability of the MapReduce paradigm while maintaining the
security of 2PC.}

\tr{A key challenge in this approach} is that GCs can only implement a
restricted class of algorithms. Specifically, GCs are
data-oblivious: they lack control flow, \tr{meaning their execution path 
does not depend on input values, and they produce fixed-size outputs 
regardless of input}. 
\tr{These constraints make standard sorting, filtering, and 
merging algorithms unsuitable, as they typically involve branching
(based on comparisons) and variable-length outputs.  Instead, \sys\
employs data-oblivious algorithms. In order to avoid overly
conservative circuit sizes (e.g., a filtering circuit that assumes
each row of a table can match the predicate when this is known to be
impossible in a given query), \sys\ limits (and pads as needed) output
to a fixed size $d$, an additional query parameter that represents a
known upper bound on possible output sizes.  Below, we describe the
core data-oblivious algorithms used in \sys\ and how they are combined
to enable efficient query execution in 2PC.  }
	
\sys\ relies on the following standard data-oblivious algorithms
implemented as GCs:
\begin{itemize}
	\item \emph{Linear scan:} Passes once over a list to perform
	operations such as marking elements or computing a running
	total.  \item \emph{Oblivious sort:} Sorts a list. While the theoretical
	complexity is $O(n \log(n))$, practical \tr{implementations
	typically run} in $O(n(\log(n))^2)$.  \sys\ uses bitonic
	sort~\cite{Batcher1968}.
        \item \emph{Oblivious sorted merge:}
	Merges two sorted lists \tr{while preserving order and}
	retaining duplicates.  \sys\ uses bitonic merge, which has a
	complexity of
	$O(n\log(n))$~\cite{Batcher1968}.
        \item \emph{Oblivious
	compact: } Moves marked records to the end of a list, ensuring
	that unmarked records remain at the front in order.  \sys\
	uses an $O(n\log(n))$ butterfly circuit algorithm (\S 3 in
	~\cite{Goodrich2011}) for this operation.
\end{itemize}

Using these algorithms, \sys\ implements mappers and reducers as
follows:\\
\noindent{\textit{Step (i): Filtering (Map phase):}
\tr{Each mapper performs a} \emph{linear scan} to
mark records that do \textit{not} satisfy the \textsf{WHERE} \texttt{condition}
with a discard bit.

\noindent{\textit{Steps (ii) and (iii): Grouping and Aggregation (Reduce Phase).}}
\begin{itemize}
\item {First-stage reducers}: 
\tr{First, the outputs of two mappers are obliviously sorted, first by the discard bit and then by the \textsf{GROUP BY} column (\texttt{col2}). This pushes discarded records to the end while ensuring that valid records are grouped together} in sorted order of \texttt{col2}. 
\tr{Next, a linear scan computes running aggregates for each unique \texttt{col2} value and marks duplicate records for removal.}
\tr{Finally, an oblivious compact operation moves discarded records to the end, leaving a sorted list of unique values of \texttt{col2} with their corresponding aggregates. The output is then padded or truncated to a fixed length $d$.}
\item Subsequent-stage reducers: 
\tr{Unlike first-stage reducers, they receive two sorted lists as input;
therefore,} they use the more efficient oblivious sorted merge \tr{to
combine records before computing final aggregates.}
\end{itemize}

The complexity of a mapper, 1st-stage reducer, and subsequent-stage
reducer is $O(c)$, $O(c (\log(c))^2)$ and $O(d \log(d))$,
respectively, where $c$ is the chunk size. The chunk size
$c$ \tr{significantly impacts} query latency: Larger
chunks \tr{increase the cost of individual} mappers and 1st-stage
reducers but reduce the total number of mappers and reducers.  In
practice, we determine the chunk size empirically to minimize query
latency; our experiments use $c = 10k$.

\paragraph{Data-dependent FGA queries.}
When a query accesses a table using a private index or when the rows
accessed by a FGA query depend on the \tr{results of a} previous
query, the DB access pattern can reveal information about the
underlying data.  To prevent such leaks, data-dependent FGA
queries \tr{operate on} previously shuffled tables
(\cref{para:random-shuffling}). Additionally, the number of records
read must \tr{remain independent of} prior query results. To this
end, \sys\ introduces dummy record reads,
\tr{ensuring that the query always processes a fixed number of records, 
regardless of actual data dependencies.}

\section{Ingress processing}
\label{sec:ingress}

Many analytics applications, such as epidemic and financial analysis,
rely on spatio-temporal data.  In these \tr{scenarios, queries often
  target specific} space-time regions. If the space-time \tr{region
  selected by a query is not} sensitive, we can \tr{index data by
  \textit{public} space-time attributes, which allows for more
  efficient query execution, because data can be filtered outside
  2PC.}  To this end, \sys\ creates materialized views indexed by
public attributes, such as coarse-grained space-time information.
\tr{While \sys\ is a read-only analytics platform, it supports
  appending new data in an order consistent with} public
attributes. The ingress processing pipeline enables incremental data
uploads from multiple data sources \tr{while producing} materialized
views securely and efficiently.
This process is implemented in 2PC by dedicated DPs called Ingress
Processors (IPs).

\paragraph{Batch append.} Data sources upload new data in batches,
which are padded to conceal \tr{the actual volume of data appended}.
\tr{Prior to upload, }data sources may: (i) pre-partition data based
on public attributes, (ii) pre-join related datasets, or (iii)
pre-filter data \tr{to reduce the size of} the upload.  Once uploaded,
batches become immutable and queryable.  Before appending a batch to
\sys's DB, the IPs may \tr{pre-process it} in 2PC.  \tr{This
  pre-processing can include buffering, grouping, joining, or sorting
  operations to produce materialized views that facilitate efficient
  query processing}. The specific operations depend on the
\tr{database scheme and} application.  For example, in
\cref{sec:epidemics}, we discuss \tr{ingress processing for our
  epidemic analytics scenario}.
Note that data sources can contribute data to 
one or more query classes. \sys's IP pairs are query-class specific.

\paragraph{Ingress security and integrity.}
Data sources upload batches to IPs using session keys, \tr{ensuring that sources
remain anonymous to} the IPs. Additionally, batch padding is only
revealed within 2PC, \tr{preventing any} single IP party \tr{from} determining the
actual \tr{batch size}. For added privacy, data sources
may randomly split a single batch into multiple smaller batches, and upload
them with different session keys. \tr{Data sources} may also use VPNs or Tor to
\tr{hide} their network addresses during uploads.

\paragraph{Storing data in garbled form.}
As discussed in \cref{subsec:fga_queries}, computing MACs in 2PC 
is expensive. 
\tr{To avoid this overhead,} \sys\  stores data directly in garbled form.
\tr{While this approach increases storage size} (as each bit is garbled into a 128-bit value), 
it significantly reduces computational overhead by eliminating both
the need for IPs to compute MACs in 2PC and 
the need for DPs to verify MACs \tr{during query execution.}
However, \tr{MAC verification remains necessary at the ingress stage,} where IPs must verify the MACs added by data sources. This step is unavoidable \tr{to ensure data integrity.}
The optimization can be generalized \tr{across multiple}
query classes using separates secrets to garble circuits for each class.
Specifically, each IP pair garbles every datum once for the query class it manages,
and each DP pair has access only to the secret \tr{for its designated query class.}

\section{Epidemic analytics scenario}
\label{sec:epidemics}

\subsection{Database}
\label{subsec:epidemics_db}
\tr{As discussed in~\cref{subsec:eva_scenario},} epidemic analytics
operates on time series of individual pairwise contacts between user
devices. This data may be collected, for instance, from smartphone
apps that track GPS coordinates and pairwise Bluetooth
encounters~\cite{Manweiler2009, Tsai2019}, or from a combination of
personal devices and Bluetooth beacons placed in key
locations~\cite{appleGoogleContactTracing,pancast}. \tr{The data is
  publicly indexed by the coarse-grained spatio-temporal region in
  which encounters occured.}

Below, we describe the records and materialized views used in our
epidemic analytics scenario.  Each view contains a variable number of
records, \tr{structured as shown} in \cref{fig:tab:appendix}.

\begin{figure}[h]%
	\setlength\tabcolsep{1.5pt}
		\scalebox{0.8}{
	\begin{tabular}{@{}c | c | c | c | c | c | c | c | c|@{}}
		\cline{2-9}
		\underline{View} & \multicolumn{8}{c|}{Record fields}\\
		\hhline{~========}
		$T_E$ & \textit{space-time-region} & eid & did1 & did2 & time & duration & val & \ldots \\
		\hhline{~========}
		$T_P$ & \textit{epoch} & did1, time & did2 & duration & prev & next & val & \ldots \\
		\cline{2-9}
	\end{tabular}
	\setlength\tabcolsep{6pt}}%
	\caption{Materialized views used in epidemic analytics. The first column is a public index.}
	\label{fig:tab:appendix}
\end{figure}

\tr{Each query class has its own dedicated materialized views}, which
includes encounters identified by an encounter ID ($eid$).  If two
data sources share an encounter, they submit the same $eid$, along
with their \tr{respective} device IDs $did$.  However, only mutually
confirmed encounters---those reported by both peers with consistent
\tr{attributes}---are \tr{considered} for analytics.  \tr{To enforce
  this}, the IPs perform a join on individual encounter reports and
add a \textit{validity} attribute (\textit{val}) to indicate whether an encounter
  is confirmed
(\cref{subsec:epidemics_ip}). 

Additionally, the records may include:
$time$, a fine-grained timestamp \tr{marking the start of an encounter},
the encounter $duration$, and any additional metadata
(e.g., fine-grained location or signal strength).

Records in $T_P$ are
indexed by $did1$ and $t$, where each entry contains the index of
$did1$'s previous ($prev$) and next ($next$) encounter.  \tr{This
  structure supports trajectory-based queries on the graph of
  encounters, such as}
q3 in \cref{fig:queries-eval} and q4 in \cref{fig:queries-eval-appendix}.

\paragraph{Space-time views.} Space-time views
store encounter records ($T_E$). These materialized views support
queries that select data by space-time region; they are indexed by a
\textit{public}, coarse-grained space-time index.
\tr{When uploading a batch, a data source assigns it to a specific space-time region.}
\tr{IPs aggregate batches from multiple sources and generate space-time views for each region.}
\tr{This approach significantly accelerates queries that target
  specific spatio-temporal regions, as DPs can locally select relevant
  records outside of 2PC.}

The granularity of space-time regions is determined by public data on
population density and mobility patterns at a given
location, day of the week, and time of the day, 
\tr{to ensure a balanced distribution of record across regions.}
\tr{Each region is padded to a fixed nominal size. Regions that are
  expected to be empty (e.g., open sea or remote wilderness) can be
  omitted.}

\paragraph{Shuffled encounter views.}
As discussed in \cref{para:random-shuffling}, space-time views can be
used for queries with no secret-dependent access \tr{patterns}. For
other queries, \sys\ uses the ODR scheme with shuffled encounter
views. In the context of our epidemic analytics scenario, such views
allow executing data-dependent FGA queries on $T_P$.
\tr{
$T_P$ contains encounters privately indexed by individual device
IDs and the times of the encounter reports (did1, time). The lookup 
key of each record is encrypted, and the mapping between a record, 
its lookup key, and its position within the view can only be reconstructed in
2PC.  
Furthermore, $T_P$ also has a public index: a conservative, 
coarse-grained public \textit{time} attribute}, which we call an \textit{epoch}.
}
The views are generated by the IPs \tr{incrementally, epoch-by-epoch}.
The records within each epoch
are re-shuffled by DPs after each use in a query.

\subsection{Ingress Processing}
\label{subsec:epidemics_ip}

We now describe \sys's ingress processing for the epidemic analytics
scenario. As mentioned in \cref{subsec:epidemics_db}, IPs determine
whether an encounter \tr{is mutually confirmed and, if so,} mark it as
valid. Below, we provide a detailed breakdown of this process.
Before uploading data, each data source 
\begin{compactitem}
\item partitions its encounters into batches per space-time region;
\item sorts each batch by encounter ID ($eid$); 
\item randomly pads the batch to hide the actual number of encounters;
\item secret-shares and MACs the batch following the protocol in \cref{subsec:macs}; and,
\item uploads the batches to the two IPs using session keys for unlinkability.
\end{compactitem}

Each IP pair continuously receives batches from data sources
and \tr{temporarily} stores them outside 2PC.
Then, \tr{at scheduled intervals} (e.g., once per day), they \tr{execute a 2PC batch processing step}, which
\begin{compactitem}
	\item consumes all per-device batches for a given space-time region;
	\item reconstructs the batches from their secret shares;
	\item verifies batch MACs to ensure data integrity;
	\item merges the (sorted) batches into a space-time region buffer using
	oblivious sorted merge~\cite{huang2012private}; and,
      \item truncates or pads the buffer to the expected space-time
        region size.
\end{compactitem}

Sorting places padded entries at the end of the buffer, so that
padding is truncated first.
\tr{Once the batches are merged,} the IPs confirm encounters within
2PC.  For each $eid$, they \tr{check if both peers reported an
  encounter with the same time and location; if the reports match, the
  valid bit is set.}
This step requires a linear scan of the space-time region.  Finally,
each IP appends the space-time region to the corresponding tables and
materialized views, and shuffles records as needed.

\subsection{Evaluation of Ingress Processing}
\label{sec:eval-ingress}

\begin{figure}[h]
	\centering
	\includegraphics[width=0.45\textwidth]{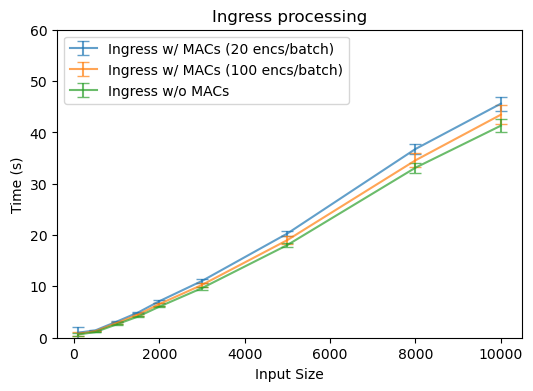}
	\caption{Cost of ingress processing on a single space-time region as a function of the region size.}
	\label{fig:eval-ingress-runtime}
\end{figure}

\cref{fig:eval-ingress-runtime} shows the runtime cost of processing a
single space-time region as a function of region size. 
The IPs independently execute all five steps
described in \cref{subsec:epidemics_ip}, running in parallel.

The two lines at the top in \cref{fig:eval-ingress-runtime} depict the
average processing time for generating a single space-time region in
$T_E$, using batch sizes of 20 or 100 encounters per upload,
respectively. The cost scales slightly super-linearly, primarily due
to the sorting and merging operations required. To quantify the MAC
overheads, we also report the cost of ingress processing without any
MAC operation (green line).  Note that the cost of ingress processing
without MAC verification depends on the region size but not the batch
size; in fact, the batch size impacts only the number of batches and
the time to upload each batch to the IPs via a different, secure
connection. This cost is negligible compared to the costs of 2PC
operations.

\tr{After encounter confirmation, the IPs store tables and views in garbled form, eliminating the need for MAC verification during queries.}

\paragraph{Scaling.}
The key performance metric for ingress processing is throughput, 
\tr{i.e., the number of encounters processed per hour (encs/h).}
We want to determine how many core pairs we need to keep up with the
data rate generated by a given administrative entity (e.g., country, city). 
Thus, we conduct a 5h experiment, where $m$ core pairs
generate space-time regions of 100 records each from uploads in
batches of 20 encs/batch. 
We vary $m$ from 1 to 8 core pairs and measure throughput. 

As expected, the results show perfect linear scaling, with throughput increasing from
$584k$ to $4.75M$ encs/h as the number of core pairs scales from 1 to 8. 
\tr{Extrapolating from these results, we estimate}
the number of core pairs required to handle \tr{different geographical scales
based on expected daily encounter rates (see \cref{tab:ingress-scaling}).}

\begin{table}[h]
	\centering
	\renewcommand{\arraystretch}{1.2}
	\begin{tabular}{|l|r|r|}
		\hline
		\textbf{Region (Population)} &\textbf{Encs/day} & \textbf{Core pairs} \\
		\hline
		Mid-sized country (80M) & 11.85B & 1660 \\
		Large metropolitan area (8M) & 1.6B & 226 \\
		Urban city (3M) & 600M & 86 \\
		Small town (200k) & 40M & 6 \\
		\hline
	\end{tabular}
	\caption{Estimated core pairs required for ingress processing at different population scales.}
	\label{tab:ingress-scaling}
\end{table}

These estimates indicate that ingress processing for a country
of 80M citizens \tr{can be performed in a} small-scale data center,
while a small town requires only a single machine pair.

\subsection{Evaluation of data-dependent queries}
\label{sec:eval-additional-query}

\paragraph{Additional query.} 
\begin{figure}[h]
  \centering
  \begin{tabular}{|l|}
    \hline
    \begin{tabular}{@{}l@{~}p{0.4\textwidth}@{}}
      \textbf{(q4)} & Count \#devices in \texttt{B} that encountered a
      device which \emph{previously} encountered a device in
      \texttt{A}, with both encounters in the time interval
      \texttt{[start,end]}
    \end{tabular}\\
    \begin{tabular}{@{}l@{}}
     \small \texttt{WITH TT AS} \\ 
     \small \texttt{~(SELECT * FROM $T_P$} \\
     \small \texttt{~~WHERE start$\,<\,$epoch$\,<\,$end)} \\ 
      \small \texttt{SELECT COUNT(DISTINCT(T2.did2))}\\
      \small \texttt{FROM (TT AS T1) JOIN (TT AS T2)}\\
      \small \texttt{~ON T1.did2 == T2.did1} \\
      \small \texttt{WHERE T1.did1$\,\in\,$A AND T2.did2$\,\in\,$B AND} \\
      \small \texttt{~start$\,<\,$T1.time$\,<\,$T2.time$\,<\,$end} 
    \end{tabular}\\
    \hline
  \end{tabular}
    \caption{Query q4. The selection on the
  	public attribute \texttt{epoch} is done outside 2PC using
  	$T_P$'s public index.}
    \label{fig:queries-eval-appendix}
\end{figure}

We describe q4, a data-dependent query (\cref{fig:queries-eval-appendix})
that we evaluate for latency. This query identifies the number of devices
in set \texttt{B} that indirectly encountered a device from set \texttt{A}
via an intermediate device, all within a given time
interval. This query can \tr{help to assess whether}
two epidemic outbreaks (corresponding to the sets \texttt{A} and \texttt{B}) are
connected via one intermediate hop.

We evaluate q4 under the same conditions as q3 (\cref{sec:eval-querylatency}). 
\tr{However,} unlike q3, which requires only one
traversal over the trajectories of devices in \texttt{B}, q4 requires \textit{two}:
One traversal collects all encounters of devices in
\texttt{A} moving forward in time; the second traversal collects all
encounters of devices in \texttt{B} moving backward in time. 

\tr{After collecting encounters, the algorithm:}
(i) sorts the collected encounters of $\mbox{\texttt{A}} \cup
\mbox{\texttt{B}}$ in ascending order by time, and (ii) processes the
list linearly, maintaining two Bloom filters: the first contains all
devices that have encountered a device in \texttt{A}; the second
contains all devices in \texttt{B} that have encountered a device
\emph{already} \tr{present} in the first Bloom filter. The query
result is the final size of the set in the second Bloom filter.

In our experiments, Bloom filter operations take 977.1 ms (std. dev. 37.2 ms) per encounter
and there is a one-time setup cost of 149 ms (std. dev. 5 ms). Sorting the encounter
list of $\mbox{\texttt{A}} \cup \mbox{\texttt{B}}$ has negligible cost
in comparison. Assuming $e = 336$ epochs, $n= 30$ encs/person/h,
$|\mbox{\texttt{A}}| = 10$ and $|\mbox{\texttt{B}}| = 10$ 
(same values as in q3), this encounter list has
length $m = e \cdot n \cdot (|\mbox{\texttt{A}}| +
|\mbox{\texttt{B}}|) = 201,600$. Thus, the total query latency is $149 + 34.9
\cdot e \cdot n \cdot (|\mbox{\texttt{A}}| + |\mbox{\texttt{B}}|) +
(9.1 + 977.1) \cdot m$ = 2.38 days. 
\tr{This method achieves significant savings compared to a} naïve approach that first sorts 165.9B encounter records in 2PC, before scanning them, while maintaining the same two Bloom filters.

\section{Estimating end-to-end query latency}
\label{sec:mr-extrapolation}

\newcommand{\mmr}{\textit{mmr}}
\newcommand{\mr}{\textit{mr}}
\newcommand{\machine}{\textit{mach}}
\newcommand{\cmr}{\textit{cmr}}
\newcommand{\fr}{\textit{fr}}
  
\newcommand{\dmmr}{D_{\mmr}}
\newcommand{\dmr}{D_{\mr}}
\newcommand{\dmachine}{D_{\machine}}
\newcommand{\dcmr}{D_{\cmr}}
\newcommand{\dfr}{D_{fr}}
\newcommand{\expect}{\mathbb{E}}
\newcommand{\cost}{\mbox{cost}}

\tr{We outline the method for estimating} end-to-end query latency
in 2PC under our MapReduce approach, as a function of the number of
available machine pairs. By reversing the estimation function, we can
also determine the number of core pairs required to attain a \tr{target}
query latency.

\subsection{\tr{Background}}

\tr{To build our performance model}, we start with two basic mathematical facts:

\begin{lem}\label{lemma:max:normal}
Let  $X_1, \ldots, X_n$ be $n$ i.i.d.\ random variables \tr{following a} normal
distribution with mean $\mu$ and standard deviation $\sigma$. Define $X =
\max(X_1,\ldots,X_n)$. Then, $\expect[X] \leq \mu + \sigma \sqrt{2
  \ln(n)}$.
\end{lem}

\begin{proof}
  This is a well-known result; we provide a simple proof \tr{for completeness}. Let $t > 0$
  be a parameter. Since $e^y$ is a convex function of $y$, Jensen's
  inequality gives: 
  $$e^{t\expect[X]} \leq \expect[e^{tX}]$$
  \tr{By definition of $X$, we have:}
  $$\expect[\max_i
    e^{tX_i}] \leq \expect[\sum_{i=1}^n e^{tX_i}] = \sum_{i=1}^n
  \expect[e^{tX_i}]$$
  From the moment generating function of the
  normal distribution: 
  $$\expect[e^{tX_i}] = t\mu + \frac{1}{2}\sigma^2
  t^2$$ 
  Hence: 
  $$e^{t\expect[X]} \leq n e^{(t\mu + \frac{1}{2}\sigma^2
    t^2)}$$
   \tr{Taking the natural logarithm on both sides:}
    $$\expect[X] \leq \frac{\ln(n)}{t} + \mu + \frac{t
    \sigma^2}{2}$$ 
    The function on the right is minimized for $t =
  \frac{\sqrt{2 \ln(n)}}{\sigma}$, and its minimum value is $\mu +
  \sigma \sqrt{2 \ln(n)}$, as required.
\end{proof}

In the sequel, we define $\expect[D; n]$ as the expected 
maximum of $n$ i.i.d.\ random variables drawn from a
normal distribution $D$. Lemma~\ref{lemma:max:normal} 
\tr{then states:}
$$\expect[D; n] \leq \mu + \sigma \sqrt{2 \ln(n)}$$ where $\mu$ and
$\sigma$ are the mean and standard deviation of $D$.

\begin{lem}
For any natural number $K \geq 0$, $$\sqrt{K} + \sqrt{K-1} + \ldots +
\sqrt{1} \leq \frac{2}{3}((K+1)^{3/2} - 1)$$
\end{lem}
\begin{proof}
  \[\begin{array}{@{}lll@{}}
  &      & \sqrt{K} + \sqrt{K-1} + \ldots + \sqrt{1} \\
  & =    & \bigint_{K}^{K+1} \sqrt{K} \;\mbox{d}x + \ldots + \bigint_{1}^{2} \sqrt{1} \;\mbox{d}x \\
  & \leq & \bigint_{K}^{K+1} \sqrt{x} \;\mbox{d}x + \ldots + \bigint_{1}^{2} \sqrt{x} \;\mbox{d}x \\
  & = & \bigint_{1}^{K+1} \sqrt{x} \;\mbox{d}x \\
  & = & \frac{2}{3}((K+1)^{3/2} - 1)
  \end{array}\]
\end{proof}

\subsection{\tr{Performance model.}}
\tr{We now detail our approach for estimating end-to-end query latency in 2PC 
under the MapReduce model, given a fixed number of machine pairs. }

Let $M$ be the number of available machine pairs, with each machine
having $2C$ cores, yielding $C$ \textit{units of DualEx execution} per
machine pair, and $MC$ total units of DualEx execution \tr{across all
  machines}.  In the following, we refer to each `unit of DualEx
execution' simply as a `unit'.

For a given query, suppose the table contains $N$ records. The records
of the table are divided into \emph{chunks} of size $t$,
which are then evenly distributed among the $MC$ units. Each unit
\tr{processes}
$N_c$ chunks, where:
$$N_c = \frac{N}{tMC}$$
The best \tr{choice} of $t$ is determined empirically: If $t$ is too small,
setup costs dominate, as each mapper and 1-stage reducer
does little work. If $t$ is too \tr{large}, the state of all parallel
cores on a machine may not fit in memory. For our
evaluation, we set $t = 10,000$ records.

\paragraph{\tr{Query execution stages.}}
A query executes in four stages. 
Each stage consists of two symmetric DualEx computations running in parallel.
\tr{Since the intermediate result of each stage until the final one is in garbled form,
only the final stage runs the DualEx equality check.}
\begin{enumerate}
\item \textit{Unit-level processing. }
  Each unit on a machine pair maps two chunks and reduces them to an
  intermediate result. 
  \tr{Each unit then alternates mapping and reducing until all $N_c$ 
  	chunks are processed into a single intermediate result.
  	This computation runs in parallel across $MC$ units.
  }
  Let $\dmmr$ be the latency distribution of the first two map +
  reduce operations on a unit. Let $\dmr$ be the latency
  distribution for each subsequent map + reduce step.
\item \textit{Machine-pair reduction}. Each machine pair
  reduces its $C$ unit results into a single result using a balanced reduction
  tree. This happens independently on each machine pair in parallel. Let
  $\dmachine$ be the latency distribution of this step on one machine.
\item \textit{Cross-machine reduction.} The intermediate results from all machine
  pairs are further reduced using a balanced reduce tree, 
  until only two intermediate result remain on a single machine pair. 
  \tr{At each level, pairs of machines reduce their results in parallel, 
  with one machine sending its result to another at the end.
  We define a cross-machine reduction unit (cmru) as one such step.
  Let $\dcmr$ be the latency distribution of one \textit{cmru} operation.
  }
\item \textit{Final reduction.} 
	\tr{The last machine pair performs a final reduction, including the DualEx equality check.}
	Let $\dfr$ be the latency distribution of this final step.
\end{enumerate}

\paragraph{\tr{Latency estimation.}}
\tr{
We estimate $\dmmr$, $\dmr$, $\dmachine$, $\dcmr$, and $\dfr$ by measuring their means ($\mu$) and standard deviations ($\sigma$) over test queries on a small number of machines. Each latency distribution is modeled as normal with observed $\mu$ and $\sigma$.
In practice, stages do not synchronize until cross-machine reduction: if a machine pair finishes stage 1, it moves on to stage 2 without waiting for others.
However, for conservative estimation, we assume full synchronization at the end of each stage.
This ensures an upper bound on query latency, which is easier to compute.
Let $\cost_1, \cost_2, \cost_3, \cost_4$ denote the expected latencies of the four stages under this conservative model. We estimate each separately:
}

\begin{enumerate}
\item \textit{Estimating $\cost_1$.}
The latency on each unit follows a normal distribution given by
$D_1 = \dmmr + (N_c - 2) \dmr$. By standard normal distribution
properties:
\[
\begin{array}{lll}
  \mu_1    & = & \mu_{\mmr} + (N_c - 2) \mu_{\mr} \\
  \sigma_1 & = & \sqrt{\sigma_{\mmr}^2 + (N_c - 2) \sigma_{\mr}^2}
\end{array}
\]
\tr{Since $MC$ units run in parallel:}
\[ \cost_1 = \expect[D_1 ; MC] \leq \mu_1 + \sigma_1 \sqrt{2 \ln(MC)} \]

\item \textit{Estimating $\cost_2$.}
Each machine pair performs a local reduction with latency $\dmachine$. 
Since $M$ machines run in parallel:
\[ \cost_2 = \expect[\dmachine; M] \leq \mu_{\machine} + \sigma_{\machine} \sqrt{2 \ln(M)} \]

\item \textit{Estimating $\cost_3$.}
\tr{This step follows a tree structure with $K$ levels}, where:
\[ K = \lceil \log_2(M) \rceil \]
\tr{At each level, pairs of machines reduce their results, 
decreasing from $M/4$ pairs at level 1 to a single pair at level $K$. 
The total latency is:}
\[
\begin{array}{lll}
  \cost_3 & \leq & \displaystyle\sum_{i = 2}^{K} \expect[\dcmr; \lceil M/(2^i) \rceil]  \\
  & \leq & \displaystyle\sum_{i = 2}^{K} \left(\mu_{\cmr} + \sigma_{\cmr} \sqrt{2 \ln \left(\lceil M/(2^i) \rceil\right)}\right) \\
  & =    & (K-1)\mu_{\cmr} + \sigma_{\cmr} \sqrt{2 \ln 2} \displaystyle\sum_{i = 2}^{K} \sqrt{\left(\lceil M/(2^i) \rceil\right)} \\
  & \leq & (K-1)\mu_{\cmr} + \sigma_{\cmr} \sqrt{2 \ln 2} \displaystyle\sum_{i = 1}^{K-2} \sqrt{i} \\
  & \leq & (K-1)\mu_{\cmr} + \frac{2}{3} \sigma_{\cmr} \sqrt{2 \ln 2} \left((K-1)^{3/2} -1\right) \\
\end{array}
\]

\item \textit{Estimating $\cost_4$.}
Since this step occurs on a single machine pair:
\[ \cost_4 = \mu_{\fr} \]
\end{enumerate}

The final upper bound on query latency is $\cost_1 + \cost_2 + \cost_3 + \cost_4$.

\section{Proofs of Security}
\label{app:proofs}

In this section, we prove the security of \sys's MAC-then-Encrypt (MtS) mechanism and its modified
DualEx protocol.

In \S\ref{app:preliminaries}, \tr{we present a secure secret sharing scheme that integrates MACs, along with a formal proof of its security.}
In \S\ref{app:sss-construction}, we combine MACs and secret sharing to give a construction of a secure secret sharing scheme and a proof of its security. 
In \S\ref{app:asym-dualex}, we extend DualEx, a symmetric 2PC protocol with one-bit leakage, to support asymmetric outputs, %
\tr{then provide a security proof for this modified protocol, which we refer to as asymmetric DualEx.}

\subsection{Notation and preliminaries}
\label{app:preliminaries}

For any deterministic functionality $\H$, let $\H^-$ be the functionality that computes $\H$ with one-bit leakage. Specifically, for a malicious party $P_i$, the two-party functionality $\H^-$ outputs the same result as $\H$, along with a one-bit leakage function $\ell$ chosen by $P_i$, applied to the other party's input~$x$ (\tr{where} $\lvert \ell(x) \rvert = 1$). For a randomized functionality, $\ell$ may depend on the other party's input and on the randomness used by the functionality.

Throughout this appendix, we assume that certain algorithms can output a special symbol $\perp$. Our security proofs consider probabilistic polynomial-time (\ppt) adversaries $\A$.
\tr{Following standard cryptographic conventions, we prove the security of our schemes by showing that}
the \emph{advantage} (defined separately for every property) of a \ppt~adversary $\A$ is upper-bounded by a \emph{negligible function}:

\begin{definition}[Negligible Function]
	A function $f$ is negligible if for all $c \in \mathbb{N}$ there exists an $N \in \mathbb{N}$ such that $f(n) < n^{-c}$ for all $n > N$.
\end{definition}

\subsection{Message authentication codes (MACs)}\label{app:mac}

We use MACs to ensure integrity of data \tr{uploaded by data sources to \sys}.
Each \tr{source} secret shares its data before sending one share to each of the two TEEs. 
\tr{However, standard additive secret sharing is inherently malleable. Thus, we employ MACs to provide integrity verification.}

\newcommand{\vrfy}{\ensuremath{\mathsf{Verify}}}
\newcommand{\tg}{\ensuremath{\mathsf{Mac}}}
\newcommand{\kgen}{\ensuremath{\mathsf{KeyGen}}}

\begin{definition}[MAC]
	A MAC \tr{scheme} is a triple of polynomial-time algorithms $M = (\kgen, \tg, \vrfy)$ such that:
	\begin{itemize}
		\item $\kgen$ takes as input a security parameter $1^\kappa$ and outputs a random key $k \sample \mathcal{D}_\kappa$
		\item $\tg$ takes as input a key $k \in \mathcal{D}_\kappa$ (associated with a security parameter $\secparam$) and a message $m \in \mathcal{D}_m$ and outputs a tag $t$.
		\item $\vrfy$ takes as input a key $k \in \mathcal{D}_\kappa$, a message $m \in \mathcal{D}_m$, and a tag $t$ and outputs either $\top$ or $\perp$.
	\end{itemize}
	
	For correctness, we require that for all $m \in \mathcal{D}_m$ and $k \in \mathcal{D}_\kappa$, $\vrfy(k, m, \tg(k, m)) = \top$.
\end{definition}

\paragraph{Notation.} \tr{We} define $\tg_k(m) := \tg(k, m)$ and $\vrfy_k(m, t) := \vrfy(k, m, t)$.

\hfill\\
The security properties of MACs 
\tr{vary across different adversarial models.
In the next section, we define the relevant security notions for our setting.}

\paragraph{One-time strong unforgeability.}
Informally, one-time strong unforgeability \tr{ensures that an attacker cannot forge a valid tag for a message without knowing the secret key, even if they have already seen a valid tag for a different message}.

Let $M=(\kgen, \tg, \vrfy)$ be a MAC \tr{scheme}. For a $\ppt$ adversary $\A$, we define $\A$'s advantage with respect to $M$ as $\textsf{MAC1-sforge-adv}[\A, M] :=$
\[
\Pr\left[ \begin{array}{c}
	m \leftarrow \A(\secparam); \\
	k \leftarrow \kgen(\secparam);\\
	t \leftarrow \tg_k(m); \\
	(m', t') \leftarrow \A(t)
\end{array} : \begin{array}{c}
	(m', t') \neq (m, t) \; \bigwedge \\
	\vrfy_k(m', t') = \top
\end{array} \right].
\]

\begin{definition}[One-time Strong Unforgeability]
	$M=(\kgen, \tg, \vrfy)$ is \emph{one-time strongly unforgeable} (or \emph{one-time strongly secure}) if, for all \ppt~adversaries $\A$, there exists a negligible function $negl$ such that $\textsf{MAC1-sforge-adv}[\A,\allowbreak M] \leq negl(\kappa)$.
\end{definition}

\paragraph{One-time key authenticity.}
\tr{Traditional MAC security considers adversaries who observe message-tag pairs. However, in some settings, an attacker may also gain access to a message-key pair. We introduce \textit{one-time key authenticity}, which ensures that it is difficult for an attacker to derive a new key that still verifies} a given message-tag pair.

Let $M=(\kgen, \tg, \vrfy)$ be a MAC \tr{scheme}. For a given \ppt~adversary $\A$, we define $\A$'s advantage with respect to $M$ as $\textsf{MAC1-kauth-adv}[\A, M] :=$
\[
\Pr\left[ \begin{array}{c}
	m \leftarrow \A(\secparam);\\
	k \leftarrow \kgen(\secparam);\\
	t \leftarrow \tg_k(m);\\
	k' \leftarrow \A(k)
\end{array} : \begin{array}{c}
	k' \neq k \; \bigwedge \\
	\vrfy_{k'}(m, t) = \top
\end{array} \right].
\]

\begin{definition}[One-time Key Authenticity]
	A MAC $M=(\kgen, \tg, \vrfy)$ is \emph{one-time key authentic} if, for all \ppt~adversaries \A, there exists a negligible function $negl$ such that \textsf{MAC1-kauth-adv}$[\A, M] \leq negl(\kappa)$.
\end{definition}

\paragraph{One-time non-malleability.}
\tr{This property prevents an adversary from modifying both the message and the key by applying additive shifts, in a way that still allows verification using against a fixed, known tag.}

Let $M=(\kgen, \tg, \vrfy)$ be a MAC \tr{scheme}. For a given \ppt~adversary $\A$, we define $\A$'s advantage with respect to $M$ as $\textsf{MAC1-nmall-adv}[\A, M] :=$
\[
\Pr\left[ \begin{array}{c}
	m \leftarrow \A(\secparam);\\
	k \leftarrow \kgen(\secparam);\\
	t \leftarrow \tg_k(m);\\
	(\Delta_m, \Delta_k) \leftarrow \A(t)
\end{array} : \begin{array}{c}
	(\Delta_m, \Delta_k) \neq (0, 0) \; \bigwedge \\
	\vrfy_{k+\Delta_k}(m+\Delta_m, t) = \top
\end{array} \right].
\]

\begin{definition}[One-time Non-malleability]
	$M=(\kgen, \tg, \vrfy)$ is \emph{one-time non-malleable} if, for all \ppt~adversaries \A, there exists a negligible function $negl$ such that \textsf{MAC1-nmall-adv}$[\A, M] \leq 
	negl(\kappa)$.
\end{definition}

\paragraph{Privacy.}
Let $M=(\kgen, \tg, \vrfy)$ be a MAC \tr{scheme}. For a given $\ppt$ adversary $\A$, we define $\A$'s advantage with respect to $M$ as $\textsf{MAC-priv-adv}[\A, M] :=$
\[
\left|\Pr\left[ 
\begin{array}{c} 
	(m_0, m_1) \leftarrow \A(\secparam);\\
	k \leftarrow \kgen(\secparam);\\
	b \sample \bin;\\
	t_b \leftarrow \tg_k(m_b)
\end{array} : \A(t_b) = b
\right] - \frac{1}{2} \right|. 
\]

\begin{definition}[Privacy]
	$M = (\kgen,\allowbreak \tg, \vrfy)$ is \emph{private} if, for all \ppt~adversaries $\A$, there is a negligible function $negl$ such that $\textsf{MAC-priv-adv}[\A, M] \leq negl(\kappa)$.
      \end{definition}

\subsubsection{Secure secret-sharing}
\label{app:sss}

\tr{Standard secret-sharing schemes ensure privacy but do not inherently guarantee the integrity of reconstructed data. That is, while individual shares do not reveal any information about the original secret, they do not prevent a malicious party from modifying their share to influence the reconstructed value.
In this section, we define an additional security property, \emph{authenticity}, which ensures that any tampering with a share is detected and results in failed reconstruction. A secret-sharing scheme that is both private and authenticated is called \emph{secure}; such a scheme can be constructed using MACs (\S\ref{app:sss-construction}).
}

\begin{definition}[Secret-sharing Scheme]
A (two-party) secret-sharing scheme is a pair of polynomial-time algorithms $\Sigma = (\share, \recon)$ defined as follows:
\begin{itemize}
    \item $\share$ takes as input a security parameter $\secparam$ and a value $x \in \mathcal{D}_\kappa$ associated with~$\kappa$ (e.g., $\mathcal{D}_\kappa = \bin^\kappa$) and outputs two shares ($\sh_1, \sh_2$). We assume $\kappa$ is implicit in each share.
    \item $\recon$ takes two shares as input and outputs either $y \in \mathcal{D}_\kappa$ (successful reconstruction) or $\perp$ (invalid reconstruction.).
\end{itemize}

For correctness, we require that for all $\kappa$ and $x \in \mathcal{D}_\kappa$, $\recon(\share(\secparam, x)) = x$.
\end{definition}

\paragraph{Privacy.}
Informally, privacy means that a single share reveals nothing about the secret. Given a share of one of two values $x_0$ or $x_1$, an adversary $\A$ cannot distinguish which value was shared. (Note: $x_i$ denotes an entire value, not a share.)

For a \ppt~adversary $\A$, we define $\A$'s \emph{advantage} with respect to $\Sigma$ as $\textsf{SS-priv-adv}[\A, \Sigma] :=$
\[
    \left|\Pr\left[ \begin{array}{c} 
        (x_0, x_1, i) \leftarrow \A(\secparam); \\
        b \sample \bin; \\
        (\sh_{b,1}, \sh_{b,2}) \leftarrow \share(\secparam, x_b) 
    \end{array} 
    : \A(\sh_{b,i}) = b\right] - \frac{1}{2} \right|. 
\]

\begin{definition}[Privacy]
    A secret-sharing scheme $\Sigma = (\share, \recon)$ is \emph{private} if, for all \ppt~adversaries $\A$, there is a negligible function $negl$ such that $\textsf{SS-priv-adv}[\A, \Sigma] \leq negl(\kappa)$.
\end{definition}

\paragraph{Authenticity.}
Informally, authenticity guarantees that any modification to a share will \tr{cause reconstruction
to fail} with high probability (i.e., output $\perp$).

For a \ppt~adversary $\A$, we define $\A$'s advantage with respect to $\Sigma$ as $\textsf{SS-auth-adv}[\A, \Sigma] :=$
\[
    \Pr\left[ \begin{array}{c}
        (x, i) \leftarrow \A(\secparam); \\
        (\sh_1, \sh_2) \leftarrow \share(\secparam, x);\\
        \sh_i' \leftarrow \A(\sh_i)
    \end{array} : \begin{array}{c}
        \sh_i' \neq \sh_i \; \bigwedge \\
        \recon(\sh_i', \sh_{3-i}) \neq \perp
    \end{array} \right].
\]

\begin{definition}[Authenticity]
    A secret-sharing scheme $\Sigma = (\share, \recon)$ is \emph{authenticated} if, for all \ppt adversaries $\A$, there is a negligible function $negl$ such that $\textsf{SS-auth-adv}[\A, \Sigma] \leq negl(\kappa)$.
\end{definition}

\begin{definition}[Security]
    A secret-sharing scheme $\Sigma = (\share, \recon)$ is \emph{secure} if it is both private and authenticated.
\end{definition}

\subsection{Secure secret sharing construction}\label{app:sss-construction}

There are multiple ways to construct a secure secret-sharing scheme $\Sigma_M = (\share_M, \recon_M)$ using a MAC $M = (\kgen,\allowbreak \tg, \vrfy)$ and base (non-authenticated) secret-sharing scheme $\Sigma = (\share, \recon)$. Here, we introduce MAC-then-Share (MtS), the construction we use in \sys\, and analyze its security \tr{properties}.

\begin{definition}[MAC-then-share (MtS)]
    Let $M = (\kgen,\allowbreak \tg, \vrfy)$ be a MAC and $\Sigma = (\share, \recon)$ a secret-sharing scheme. The MtS secret-sharing scheme $\Sigma_M = (\share_M, \recon_M)$ is defined as follows:
    \begin{itemize}
        \item $\share_M$ takes as input a security parameter $\secparam$ and a value $x \in \mathcal{D}_\kappa$ associated with $\kappa$. It generates a MAC key $k$ using $\kgen$ and computes one tag $t := \tg_k(x)$ on $x$. Then, it secret-shares both $x, k$ in $(x_1, x_2) \leftarrow \share(x)$ and $(k_1, k_2) \leftarrow \share(k)$. Finally, it outputs the two shares $\sh_1$, $\sh_2$ with $\sh_i := (x_i, k_i, t)$.
        \item $\recon_M$ takes as input two shares. If the tags match, it reconstructs $y := \recon(x_1, x_2)$ and $k' := \recon(k_1, k_2)$. Then, if MAC verification succeds ($\vrfy_{k'}(y, t) = \top$), it returns $y$; otherwise it returns $\perp$.
    \end{itemize}
\end{definition}

\tr{What must the MAC and base secret-sharing scheme satisfy to imply the security of the above MtS scheme? In Theorem~\ref{thm:mts}, we show that MtS is secure if the MAC satisfies all four properties from \S\ref{app:mac}}.

\begin{thm}\label{thm:mts}
    If $M$ is a one-time strongly unforgeable, key authentic, non-malleable, and private MAC \tr{scheme}, and $\Sigma$ is an additive secret-sharing scheme, then the MtS scheme $\Sigma_M$ is secure.
\end{thm}

\begin{proof}
    The privacy of $\Sigma_M$ follows directly from \tr{the} privacy of $M$ (as $\tg_k(x)$ reveals nothing about $x$) and $\Sigma$ (as $x_i$ reveals nothing about $x$).
    
    The authenticity of $\Sigma_M$ is a bit more unwieldy. We prove the contrapositive: if $\Sigma_M$ is not authenticated, then \tr{at least one of the following must be false:} (i) $M$ is not strongly unforgeable, (ii) $M$ lacks key authenticity, or (iii) $M$ is malleable. 
    As before, we do this by reduction of an adversary $\A$ with $\textsf{SS-auth-adv}[\A, \Sigma_M]$ nonnegligible to three adversaries for each of the three corresponding games, and show that at least one of them has nonnegligible advantage in its game.

    \begin{enumerate}
    \item We construct an adversary $\A_\textsf{sforge}$ for the game $\textsf{MAC1-sforge}_{\A_\textsf{sforge}, M}$. $\A_\textsf{sforge}$ receives $x,i$ from $\A$ and constructs a MtS triple $\sh_i := (x_i, k_i, t_i)$ as follows: it sends $x$ to its game to get $t := \tg_k(x)$, picks a random $k_i \in \D_\kappa$, and runs $\share$ on $x, t$ to get $x_1, x_2, t_1, t_2$. Now $\A_\textsf{sforge}$ runs $\A$ on $\sh_i$ to get $\sh_i' := (x_i', k_i', t_i')$ and returns $(x', t')$, where $x' \leftarrow \recon(x_i', x_{3-i})$ and $t' \leftarrow \recon(t_i', t_{3-i})$.
    
    \item We construct an adversary $\A_\textsf{kauth}$ for the game \textsf{MAC-kauth}$_{\A_\textsf{kauth}, M}$. $\A_\textsf{kauth}$ is given $x,i$ by $\A$ and sends $x$ to its game, which sends back a key $k$. It uses this key to construct a MtS triple $\sh_i := (x_i, k_i, t)$, computing $t := \tg_{k_i}(x)$ and running $\share$ on $x, k$ to get $x_1, x_2, k_1, k_2$. Now $\A_\textsf{kauth}$ runs $\A$ on $\sh_i$ to get $\sh_i' := (x_i', k_i', t')$ and returns $k' \leftarrow \recon(k_i', k_{3-i})$.
    
    \item We construct an adversary $\A_{\textsf{nmall}_M}$ for the MAC non-malleability game.$\A_{\textsf{nmall}_M}$ is given $x,i$ by $\A$ and sends $x$ to its game to receive $t$. It constructs a MtS triple $\sh_i := (x_i, k_i, t_i)$ by choosing a random $k \in \D_\kappa$ and running $\share$ on $x, k, t$. Now $\A_{\textsf{nmall}_M}$ runs $\A$ on $\sh_i$ to get $\sh_i' := (x_i', k_i', t')$. It computes $x' \leftarrow \recon(x_i', x_{3-i}), k' \leftarrow \recon(k_i', k_{3-i})$, and returns $(x', k')$.
    \end{enumerate}
    
    We now analyze the winning probability of each of the three adversaries. By our assumption, $\A$ wins its game with nonnegligible probability, implying $t' = t$. If $\A$ returns $\sh_i'$ such that $k_i' = k_i$ with nonnegligible probability, then $\A_\textsf{sforge}$ has nonnegligible advantage and we are done. If not, then with nonnegligible probability, $\A$ returns $\sh_i'$ with $k_i' \neq k_i$. If $x_i' = x_i$ a nonnegligible fraction of the time, $\A_\textsf{kauth}$ has nonnegligible advantage in its game, and we are done. Otherwise, $x_i' \neq x_i$ and $\A_{\textsf{nmall}_M}$ has nonnegligible advantage in the non-malleability game for $M$: it has a pair $(\Delta_m := x_i' - x_i, \Delta_k := k_i' - k_i)$ such that $\vrfy_{k_i' + k_{3-i}}(x_i' + x_{3-i}, t) = \vrfy_{(k_i + k_{3-i}) + \Delta_k}((x_i + x_{3-i}) + \Delta_m, t) = \top$.
    
    Thus, if $M$ is authenticated, it is not strongly unforgeable, lacks key authenticity, or is malleable, all of which contradict our assumptions. 
\end{proof}

Note that the proof also holds for the XOR secret sharing scheme by substituting $+,-$ for $\oplus$.

\paragraph{Concrete Construction.} \sys~uses KMAC256\cite{kmac-nist}, a SHA3-based MAC standardized by NIST. The scheme can be abstracted as follows, where $||$ denotes concatenation. (We omit the parameters which are constant and public in \sys; for these parameters, see \S A in \cite{kmac-nist} and \S 6.1 in \cite{sha3-nist}.) 

\def\sha3{\textsf{SHA3}}
\begin{description}
    \item[] $\kgen$: Uses \sha3's key generation algorithm.
    \item[] $\tg_k(m)$: Computes $h := \sha3(k)$ and announces it publicly. Outputs $t := \sha3(k || m)$.
    \item[] $\vrfy_k(m,t)$: Outputs $\top$ iff $\sha3(k) = h ~\wedge~ \sha3(k || m) = t$.
\end{description}
This MAC satisfies the conditions of Theorem~\ref{thm:mts} in the random oracle model (ROM):
\begin{itemize}
    \item \textbf{One-time strong unforgeability:} Due to randomness of the output $\sha3(k || m')$, \tr{it is infeasible to forge a valid tag without knowing $k$.}
    \item \textbf{One-time key authenticity:} Due to collision-resistance, if $k \neq k'$, then $\sha3(k || m) \neq \sha3(k' || m)$ with overwhelming probability.
    \item \textbf{One-time non-malleability:} Also due to collision-resistance, if $(k,m) \neq (k',m')$, then with overwhelming probability, $\sha3(k || m) \neq \sha3(k' || m')$.
    \item \textbf{privacy:} Due to randomness of the output of SHA3, \tr{no useful information about 
    $m$ is leaked.}
\end{itemize}

\tr{Since KMAC256 meets these security requirements,} \sys's MtS construction $\Sigma_M$ using $\Sigma$ as the additive secret-sharing scheme and $M$ as KMAC256 is \textit{secure} \tr{in the ROM model.}

\subsection{DualEx functionality with asymmetric outputs}\label{app:asym-dualex}

Before considering our secure secret-sharing construction within a secure two-party computation (2PC) protocol, we must first define the exact protocol we use and its non-standard security guarantees.

We use DualEx~\cite{dualex} as a building block. DualEx provides security against malicious adversaries at a cost nearly equivalent to that of semi-honest protocols. However, it comes with two caveats: (i) it only supports \emph{symmetric} two-party functionality $\F_\text{sym}$ \tr{(meaning both parties receive the same output)}, and (ii) \tr{it leaks one bit of information}.

In this section, we modify DualEx to \tr{compute \textit{asymmetric} functionalities} while preserving the same security guarantees (malicious security with one-bit leakage). This modified version, which we call \textit{asymmetric} DualEx, is used to implement the core decryption functionality, and the one-bit leakage \tr{propagates throughout} the security proofs.

Let $f : \X \times \X \rightarrow \Y$ be 
\tr{a function where the output is to be shared between the two parties, ensuring that neither party learns it in full.}
Specifically, we \tr{define} the \textit{asymmetric} functionality $\F_\text{asym}$ as follows. $\F_\text{asym}$ takes as input $x_1$ from party $P_1$ and $x_2$ from party $P_2$, and it returns a uniformly random $r$ to $P_1$ and  $f(x_1, x_2) \oplus r$ to $P_1$. When $P_1$ is honest, $r$ is chosen by $\F_\text{asym}$; when $P_1$ is malicious, $r$ may be chosen by $P_1$.

To achieve this, we construct a 2PC protocol, $\Pi$, that computes $\F_\text{asym}$ \tr{using a}
symmetric functionality $\F_\text{sym}$, defined as follows: it takes as inputs $(x_1, r_1), (x_2, r_2) \in \X \times \Y$ from the $P_1$ and $P_2$, it computes $y := f(x_1, x_2) \oplus r_1 \oplus r_2$, and returns $y$ to both $P_1$ and $P_2$. 
The 2PC protocol proceeds as follows: 
\begin{itemize}
    \item Each party $P_i$ holds input $x_i \in \X$ and generates a uniformly \tr{random} blinding value 
    $r_i \in \Y$. 
    \item The parties $P_1$ and $P_2$ submit their respective inputs $(x_1, r_1)$ and $(x_2, r_2)$ to~$\F_\text{sym}$. 
    \item Both parties $P_1$ and $P_2$ receive $y := f(x_1, x_2) \oplus r_1 \oplus r_2$.
    \item $P_1$ computes $y_1 := r_1$, while $P_2$ computes $y_2 := y \oplus r_2$.
\end{itemize}

As $y \oplus r_2 = f(x_1, x_2) \oplus r_1$, the protocol outputs $f(x_1, x_2) \oplus r_1$ to $P_2$ and the parties now hold shares of $f(x_1, x_2)$. 

Next, we show that this modified protocol maintains the security guarantees of the original DualEx.
    
\begin{thm}\label{thm:asym-dualex}
$\Pi$ securely computes $\F_\text{asym}^-$ against malicious adversaries in the $\F_\text{sym}^-$-hybrid model.
\end{thm}

\begin{proof}
Let $\A$ be a $\ppt$ adversary corrupting party $P_i$. To prove security against malicious adversaries in the $\F_\text{sym}^-$-hybrid world, we %
\tr{construct an ideal-world simulator $\Sim_i$ (where $i$ is the index of the corrupted party)}
 that simulates an execution of $\Pi$ in the hybrid world. 
 
Next, we first consider the case of a corrupted~$P_1$.

\hfill\\
\noindent
\textbf{Simulator $\Sim_1$:} $\Sim_1$ has access to the ideal functionality $\F_\text{asym}^-$, which computes $f(x_1, x_2) \oplus r$ with one-bit leakage. Given f, $\Sim_1$ operates as follows:
\begin{itemize}
    \item Receives inputs $x_1', r_1'$, and a leakage function $\ell: \X \times \Y \rightarrow \bin$ from $P_1$.
    \item Samples $y^* \sample \Y$ and converts $\ell$ into a function $\ell^*: \X \rightarrow \{0,1\}$ by letting $\ell^*(x) = \ell(x, y^* \oplus r_1' \oplus f(x_1', x))$ for all $x \in \X$. 
    \item Sends $(x_1', r_1')$ and $\ell^*$ to ~$\F_\text{asym}^-$, while the honest party sends $x_2$ to $\F_\text{asym}^-$.
    \item Receives $r_1'$ and some one-bit leakage $b^* := \ell^*(x_2)$ from $\F_\text{asym}^-$. The honest party receives $y_2^* := f(x_1', x_2) \allowbreak \oplus r_1'$.
    \item Sends $y^*, b^*$ to $P_1$.
\end{itemize}

\hfill\\
We prove indistinguishability of the following distribution ensembles:

\paragraph{Ideal experiment.} This is defined by the interaction of $\Sim_1$ with $\F_\text{asym}^-$. $\A$ outputs an arbitrary function of its view; the honest party outputs what it received from the experiment, namely $y_2^*$. Let $\ideal_{\Sim_1}(\secparam, x_1, x_2, \ell, f)$ be the joint random variable containing the output of the adversary $\Sim_1$ and the output of the honest party. 
Concretely, 
\[
    \ideal_{\Sim_1}(\secparam, x_1, x_2, \ell, f) = ((y^*, b^*), y_2^*).
\]

\paragraph{Hybrid experiment.} let $\textsc{H}_\A(\secparam, x_1, x_2, \ell, f)$ be the joint random variable containing the view of the adversary $\A$ and the output of the honest party in the $\F_\text{sym}^-$-hybrid world. Concretely, 
\[
    \textsc{H}_\A(\secparam, x_1, x_2, \ell, f) = ((y, b), y_2).
\]

\hfill\\
$\Sim_1$ perfectly simulates the view of a malicious $P_1$ in the hybrid world. Because $r_2$ is chosen uniformly at random, both $y^*$ and $y$ are distributed uniformly at random and are thus indistinguishable. Similarly, $b^* = \ell(x_2, y^* \oplus f(x_1', x_2) \oplus r_1'$ by definition of $\ell^*$.
Furthermore, by the definition of $y$, $r_2 = y \oplus f(x_1', x_2) \oplus r_1'$, so $b = \ell(x_2, r_2)$ is indistinguishable from $b^*$. Hence the joint distributions of $(y^*, b^*)$ and $(y, b)$ are indistinguishable. Finally, $y_2^*$ and $y_2$ are identical by the definition of $y_2$ and thus indistinguishable.

Next, we give a simulator $\Sim_2$ for the case of a corrupted $P_2$.

\hfill\\
\noindent
\textbf{Simulator $\Sim_2$:} $\Sim_2$ has access to the ideal functionality $\F_\text{asym}^-$ computing $f(x_1, x_2) \oplus r$ with one-bit leakage. Given $f$, $\Sim_2$ works as follows:
\begin{itemize}
    \item Receives inputs $x_2', r_2'$, and a leakage function $\ell$  from $P_2$.
    \item Forwards $x_2'$ and $\ell$ to ~$\F_\text{asym}^-$; the honest party sends $x_1$ to $\F_\text{asym}^-$.
    \item Receives in return from $\F_\text{asym}^-$ the output $z := f(x_1, x_2') \oplus r_1$ (for uniformly random $r_1$) and some one-bit leakage $b := \ell(x_1)$. The honest party receives $r_1$.
    \item Computes $y^* := z \oplus r_2'$. Send $y^*, b$ to $P_2$.
\end{itemize}

\hfill\\
We prove indistinguishability of the following distribution ensembles:

\paragraph{Ideal experiment.} This is defined by the interaction of $\Sim_2$ with $\F_\text{asym}^-$. $\A$ outputs an arbitrary function of its view; the honest party outputs its honestly computed value of $y_1$, namely $r_1$. Let $\ideal_{\Sim_2}(\secparam, x_1, x_2, \ell, f)$ be the joint random variable containing the output of the honest party and the output of the adversary $\Sim_2$. 
Concretely, 
\[
    \ideal_{\Sim_2}(\kappa, x_1, x_2, \ell, f) = (r_1, (y^*, b)).
\]

\paragraph{Hybrid experiment.} let $\textsc{H}_\A(\secparam, x_1, x_2, \ell, f)$ be the joint random variable containing the the output of the honest party and the view of the adversary $\A$ in the $\F_\text{sym}^-$-hybrid world. Concretely, 
\[
    \textsc{H}_\A(\secparam, x_1, x_2, \ell, f) = (r_1, (y, b)).
\]

\hfill\\
$\Sim_2$ perfectly simulates the view of a malicious $P_2$ in the hybrid world, since $y^* = z \oplus r_2' = f(x_1, x_2') \oplus r_1 \oplus r_2' = y$.

\hfill\\
Therefore, $\Pi$ is secure (up to one-bit leakage) against malicious adversaries in the $\F_\text{sym}^-$-hybrid model.
\end{proof}

Notice that DualEx instantiates $\F_\text{sym}^-$, thus our modified protocol is a real-world protocol that preserves DualEx's security guarantees against malicious parties: privacy up one-bit leakage and full correctness of the output.

\else
\fi

\end{document}